\documentclass[a4paper]{article}

\pdfoutput=1

\usepackage[T1]{fontenc}
\usepackage[utf8]{inputenc}
\usepackage[margin=29mm]{geometry}
\usepackage{csquotes}
\usepackage[english]{babel}
\usepackage[style=numeric-comp,sorting=none,giveninits]{biblatex}
\usepackage[font=footnotesize]{caption}
\usepackage{tikz}
\usepackage{pgfplots}
\usepackage{tikzscale}
\usepackage{amsmath}
\usepackage{amsfonts}
\usepackage{amssymb}
\usepackage{amsthm}
\usepackage[affil-it]{authblk}
\usepackage{braket}
\usepackage{bm}
\usepackage{booktabs}
\usepackage{enumitem}
\usepackage{graphicx}
\usepackage{mathrsfs}
\usepackage{mathtools}
\usepackage[binary-units]{siunitx}
\usepackage[referable,flushleft]{threeparttablex}
\usepackage[colorlinks]{hyperref}
\usepackage[capitalise]{cleveref}

\usetikzlibrary{matrix,positioning,arrows.meta}
\pgfplotsset{compat=1.14}
\usepgfplotslibrary{fillbetween}

\newtheorem{theorem}{Theorem}
\newtheorem{lemma}[theorem]{Lemma}
\newtheorem{conjecture}[theorem]{Conjecture}
\newtheorem{corollary}[theorem]{Corollary}
\theoremstyle{definition}
\newtheorem{definition}[theorem]{Definition}
\theoremstyle{remark}
\newtheorem*{remark}{Remark}

\Crefname{rule}{Rule}{Rules}

\newcommand{\rectdim}[2]{$#1 \times #2$}
\newcommand{\transpose}{\mathsf{T}}

\title{Quantum Magic Rectangles: Characterization and \\ Application to Certified Randomness Expansion}

\author[ \hspace{-1ex}]{Sean A. Adamson\thanks{\href{mailto:sean.adamson@ed.ac.uk}{\texttt{sean.adamson@ed.ac.uk}}}}
\author[ \hspace{-1ex}]{Petros Wallden\thanks{\href{mailto:petros.wallden@ed.ac.uk}{\texttt{petros.wallden@ed.ac.uk}}}}
\affil[ \hspace{-1ex}]{School of Informatics, University of Edinburgh, \protect\\ 10 Crichton Street, Edinburgh EH8 9AB, United Kingdom}

\date{}

\begin{document}

\maketitle

\begin{abstract}
    We study a generalization of the Mermin--Peres magic square game to arbitrary rectangular dimensions. After exhibiting some general properties, these rectangular games are fully characterized in terms of their optimal win probabilities for quantum strategies. We find that for $m \times n$ rectangular games of dimensions $m,n \geq 3$ there are quantum strategies that win with certainty, while for dimensions $1 \times n$ quantum strategies do not outperform classical strategies. The final case of dimensions $2 \times n$ is richer, and we give upper and lower bounds that both outperform the classical strategies. Finally, we apply our findings to quantum certified randomness expansion to find the noise tolerance and rates for all magic rectangle games. To do this, we use our previous results to obtain the winning probability of games with a distinguished input for which the devices give a deterministic outcome, and follow the analysis of C. A. Miller and Y. Shi [\href{https://doi.org/10.1137/15M1044333}{SIAM J. Comput. \textbf{46}, 1304 (2017)}].
\end{abstract}

\section{Introduction}

Quantum theory has been arguably one of the most successful scientific theories, especially in terms of accuracy of predictions and applications.
We are currently in the midst of the second ``quantum revolution'', where the ability to control quantum systems with great precision has resulted in a new wave of technological applications.
What makes quantum theory unique is the fact that our classical intuition frequently fails, and it has been proven that understanding the foundations of this theory is crucial to fully realize the possibilities it offers.
Quantum nonlocality and contextuality are two such concepts that conflict with our classical intuition, and at the same time enable one of the most interesting applications: that of device-independent cryptographic protocols.
Device-independence, first introduced by \textcite{mayers1998quantum}, is the property that allows parties to achieve cryptographic tasks---from key distribution \cite{vazirani2014fully} to certified randomness expansion \cite{colbeck2011private}, oblivious transfer \cite{kundu2020device}, and secure quantum computation \cite{gheorghiu2015robustness}---without trusting the inner workings of their own devices.

Nonlocality is frequently expressed in terms of ``guessing'' games, in which remote parties that share entanglement try to fulfill a certain winning condition.
Finding the optimal winning strategies for quantum and classical parties in these games is the key to using nonlocality for applications such as device-independent cryptography.
\textcite{mermin1990simple,peres1990incompatible} introduced one such game called the \emph{magic square game} (see details in \cref{sec:magic_square_game}).
This game has a special place in the foundations of quantum theory due to two notable properties.
Firstly, it is one of the simplest examples where quantum strategies can win with certainty (probability one) while classical strategies cannot.
This property is also referred to as \emph{quantum pseudotelepathy} \cite{brassard2005quantum} and can be used to illustrate (strong) contextuality in the spirit of the Kochen--Specker theorem \cite{kochen1975problem}.
Secondly, it is the simplest two-player game where the maximal nonlocality can be demonstrated using only Clifford computations \cite{gottesman1998theory} (preparation of Bell states and Pauli measurements).
In comparison, the CHSH game requires one player to measure in a non-Pauli basis.
The magic square game can, in principle, be used for any of the device-independent cryptographic tasks, and its performance in comparison to other games evaluated case-by-case.
Furthermore, it can be used for efficient self-testing (e.g. \cite{supic2020self}), another exciting concept made possible by nonlocality.
That is, parties can deduce from their purely classical observations the (essentially) exact quantum state they share---a property stronger than simply observing nonclassical correlations.

In this paper, we explore a generalization of the magic square game, the winning probabilities that can be achieved, what qualitative properties are preserved, and how the generalizations can be used in applications.
The specific application we focus upon is certified randomness expansion, while analysis of other device-independent cryptographic primitives is deferred to future publications.

\paragraph{Our contributions.}
We introduce a new class of nonlocal games which we call \emph{magic rectangle} games, characterize their winning probabilities, and apply the results to certified randomness expansion.
Note, however, that the term ``magic rectangle'' has been used differently in the past, to refer to observables arranged into a rectangular array \cite{harvey2008bks,saniga2012finite}.
\begin{itemize}
    \item[--] We define a generalization of the Mermin--Peres magic square game to general rectangular dimensions (\cref{def:magic_rectangle}).
    \item[--] We fully characterize the optimal winning probabilities for quantum behaviors of all these magic rectangle games (\cref{thm:characterization}).
    \item[--] In order to achieve this characterization, we first prove a number of general properties, showing that the optimal winning probabilities for any set of behaviors (local, quantum, almost quantum, or nonsignaling) are:
    (i) the same for all games of the same dimension,
    (ii) symmetric with respect to row/column exchange, and
    (iii) monotonically increasing with the dimension of the rectangle.
    \item[--] Using the known fact that the regular magic square game (which is a special case of \rectdim{3}{3} magic rectangle games) can be won for quantum strategies with certainty, we reduce the full characterization of magic rectangles to that of \rectdim{1}{n} and \rectdim{2}{n} games (\cref{thm:chara_dims}).
    We also show that the CHSH game, according to our definitions, is a \rectdim{2}{2} magic rectangle game (\cref{thm:2x2_chsh_equiv}).
    We then obtain the optimal winning probabilities for the \rectdim{1}{n} case, while we lower and upper bound the winning probabilities for \rectdim{2}{n} games.
    To upper bound the probabilities, we conjecture the almost quantum winning probability based on numerical evidence.
    As a side result, we get that \rectdim{2}{n} games with $n\geq3$ can be won with certainty using behaviors at level 1 of the NPA hierarchy (and so exhibit a version of ``pseudotelepathy''), while the quantum and almost quantum sets both give winning probabilities strictly smaller than unity (thus not exhibiting pseudotelepathy).
    \item[--] Finally, we use this characterization to analyze certified randomness expansion from magic rectangle games.
    Specifically, we show that the winning probability of an \rectdim{m}{n} game with a distinguished input (with deterministic outcomes) can be obtained from the \rectdim{(m-1)}{(n-1)} game (\cref{thm:distinguished_bound}).
    This, along with the results of \cref{thm:characterization}, allows us to determine the noise tolerance (robustness) of each of these games.
    We then follow the analysis of \textcite{miller2017universal} to get rates for certified randomness expansion using different magic rectangle games (see \cref{tab:randomness_expansion_performance}).
\end{itemize}

\paragraph{Related works.}
The magic square game was introduced by \textcite{mermin1990simple,peres1990incompatible}, while \textcite{cabello2001bell,cabello2001all} and subsequently \textcite{aravind2002bell} stated it as a two-player nonlocal game.
\textcite{aravind2004quantum} gives a nontechnical demonstration of the Mermin--Peres magic square game.
The term quantum \emph{pseudotelepathy} was first introduced by \textcite{brassard2003multi}, and the magic square game, along with many others that share the property that there exist perfect quantum (but not classical) strategies, were reviewed in \cite{brassard2005quantum}.
There are a number of generalizations of the magic square that have been considered in literature.
\textcite{cleve2014characterization} analyze quantum strategies for ``binary constraint'' games---a general class of games that contains the magic rectangles we define---and give some (weaker than our analysis) upper bounds on winning probabilities from quantum strategies.
\textcite{arkhipov2012extending} generalized the magic square and magic pentagram games to be played on hypergraphs called \emph{arrangements}, and characterized which arrangements can exhibit quantum pseudotelepathy.
\textcite{coladangelo2017robust} considered ``linear constraint'' games, focusing on the uniqueness of winning quantum strategies in order to use such games for self-testing.

To determine optimal quantum strategies, it is important to be able to check if a given experimental behavior admits a quantum model/realization.
This question is directly linked with the question of the ``degree of nonlocality'' present in quantum theory.
\textcite{navascues2007bounding,navascues2008convergent} addressed this by giving an infinite hierarchy of conditions that are satisfied by quantum behaviors, known as the NPA hierarchy.
\textcite{navascues2015almost} defined the \emph{almost quantum} set of behaviors, which is the set closest to the quantum set that arises in a ``natural'' way and is easy to check.
Sets of behaviors that are easy to handle and include the quantum set, as is the case for the levels of the NPA hierarchy and the almost quantum set, have been used successfully to bound the winning probabilities of quantum parties in many cryptographic settings---something we also exploit in this work.

Certified randomness expansion was first introduced by \textcite{colbeck2011private}.
\textcite{vazirani2012certifiable} demonstrated quantum security for an exponential expansion protocol.
Subsequently, \textcite{miller2016robust} additionally obtained cryptographic security and robustness.
\textcite{acin2016certified} reviewed efforts to design device-independent quantum random number generators (up to 2016), and included a comparison of the main protocols.
\textcite{miller2017universal} give the spot-checking protocol that we use for our analysis of certified randomness expansion, and to obtain bounds on expansion rates.
Finally, \textcite{arnon2019simple,brown2020framework} detail alternative techniques, which give better rates for the spot-checking protocol by using the entropy accumulation theorem \cite{dupuis2020entropy,dupuis2019entropy}.
These are more involved and case-specific than \cite{miller2017universal} and, thus, to give a general analysis of certified randomness for all magic rectangle games, we use \cite{miller2017universal} in our work.
Note, however, that the noise tolerance we obtain for the different magic rectangle games does not depend on the specific technique used to bound the rates, and thus applies in general.

\paragraph{Organization of the paper.}
In \cref{sec:preliminaries}, we give some background on the magic square game and different levels of correlations.
In \cref{sec:magic_rectangle_games}, we define magic rectangle games, and in \cref{sec:prelim_results} give some general results for these games.
\Cref{sec:characterization} gives a full characterization of the winning probabilities of magic rectangle games.
We then apply our results to certified randomness expansion in \cref{sec:app_randomness}, and conclude in \cref{sec:conclusion} where we discuss our results and give future directions.

\section{Preliminaries}
\label{sec:preliminaries}

\subsection{The magic square game}
\label{sec:magic_square_game}

The Mermin--Peres magic square game \cite{aravind2004quantum} consists of two players, Alice and Bob, who are not allowed to communicate during each round of the game.
This could be achieved, for example, by ensuring a spacelike separation between the two players.
Each round consists of Alice and Bob respectively being assigned a row and column of an empty \rectdim{3}{3} table uniformly at random, which they must fill according to the rules:
\begin{enumerate}[label=S\arabic*.,ref=S\arabic*]
    \item \label[rule]{rule:stand_plus_minus} Each filled cell must belong to the set $\{+1, -1\}$.
    \item \label[rule]{rule:stand_alice_prod} Rows must contain an even number of negative entries (i.e., the product of Alice's entries to any assigned row must be $+1$).
    \item \label[rule]{rule:stand_bob_prod} Columns must contain an odd number of negative entries (i.e., the product of Bob's entries to any assigned column must be $-1$).
\end{enumerate}
Neither player has knowledge of which row or column the other has been assigned, and nor does either player know what values the other has entered.
The game is won if both players enter the same value into the cell shared by their row and column.
It is clear that the optimal classical strategy succeeds with probability $8/9$ only \cite{brassard2005quantum}, and may be achieved by both players agreeing to each follow a particular configuration for their entire table before the game begins.
Strikingly, if the players are allowed to share an entangled quantum state, it is possible for them to win the magic square game with certainty \cite{mermin1990simple,peres1990incompatible}.
Such games are said to exhibit quantum \emph{pseudotelepathy} \cite{brassard2005quantum}, setting them apart from many other nonlocal games (including the CHSH game) for which optimal quantum strategies are not guaranteed to win.

A possible quantum winning strategy for the magic square allows the players to share the entangled state
\begin{equation}
\label{eq:double_bell_state}
    \ket{\Psi} = \ket{\Phi^{+}}_{1,2} \otimes \ket{\Phi^{+}}_{3,4} ,
\end{equation}
which is the product of two maximally entangled two-qubit Bell states
\begin{equation}
    \ket{\Phi^{+}}_{a,b} \equiv \frac{\ket{0}_{a} \otimes \ket{0}_{b} + \ket{1}_{a} \otimes \ket{1}_{b}}{\sqrt{2}} .
\end{equation}
That is, Alice's quantum system is composed of qubits 1 and 3, and Bob's system of qubits 2 and 4.
Depending on which row and column are assigned, the players make measurements on their respective quantum systems according to the observables given in the corresponding cells of \cref{fig:3x3_strategy}.
The outcomes of these determine the values which Alice and Bob should enter into their respective row and column to win with certainty.
\begin{figure}[hbt]
    \centering
    \includegraphics{figures/figure1.tikz}
    \caption{
        A quantum strategy for the magic square game, in which the players share the entangled state $\ket{\Psi}$ given in \cref{eq:double_bell_state}.
        Observables $X$, $Y$, and $Z$ are the Pauli spin operators, and $I$ is the identity operator.
        Measurements of Alice correspond to a row, and those of Bob to a column.
        Each row is formed of mutually commuting observables whose product is equal to $I$, and each column of mutually commuting observables whose product is $-I$.
        The eigenvalues of each observable are $+1$ and $-1$.
        These facts combined show \cref{rule:stand_plus_minus,rule:stand_alice_prod,rule:stand_bob_prod} are automatically satisfied.
        Moreover, if $O_{A}$ is any of the given observables for Alice's system, and $O_{B}$ is the corresponding observable for Bob's system, the correlation $\bra{\Psi} O_{A} O_{B} \ket{\Psi} = 1$ guarantees the players always win.
        }
    \label{fig:3x3_strategy}
\end{figure}

\Cref{fig:3x3_strategy} shows that, unlike for the CHSH game, optimal quantum strategies for the magic square game can be implemented by performing measurements of the Pauli group only.

\subsection{Levels of correlations}

We consider local measurements made on a system shared by two observers, Alice and Bob (multipartite generalizations exist, however, we will only focus on two parties, since it is the setting we consider in this work).
Alice chooses an input $x \in \mathcal{X}$ and observes a corresponding measurement output $a \in \mathcal{A}_{x}$.
Similarly, Bob chooses an input $y \in \mathcal{Y}$ and observes a measurement output $b \in \mathcal{B}_{y}$.
We may implicitly assume that inputs for Alice and Bob are distinguishable from one another, and that each output is labeled by its corresponding input.
Hence, we may write the sets of all possible outputs for Alice and Bob respectively as the disjoint unions $\mathcal{A} = \bigcup_{x \in \mathcal{X}} \mathcal{A}_{x}$ and $\mathcal{B} = \bigcup_{y \in \mathcal{Y}} \mathcal{B}_{y}$.
We refer to a fixed configuration of all probabilities $P(a,b \mid x,y)$ as a \emph{behavior}.
These behaviors can also be thought of as vectors in $\mathbb{R}^{\lvert\mathcal{A}\times\mathcal{B}\rvert}$, a convention that is particularly useful for dealing with classes of behaviors that are then mapped to sets of vectors.

Behaviors can be characterized according to properties they have, or according to what physical theories can give rise to such behaviors.
The weakest condition (and thus the most general set of behaviors) one typically imposes is that ``signaling'' should be forbidden; behaviors should not allow for superluminal communication.
A behavior is said to exhibit \emph{nonsignaling} correlations \cite{masanes2006general} if it satisfies both $P(a \mid x) =  P(a \mid x,y)$ and $P(b \mid y) = P(b \mid x,y)$, i.e., the input of one party does not influence the probability of outcomes for the other party.
Similarly, a behavior exhibits \emph{quantum} correlations if it is realizable under the laws of quantum mechanics, meaning that there exists a joint state $\ket{\psi}$ and ``local'' measurement operators $[E_x^a,E_y^b]=0$ that reproduce the behavior, i.e.,  such that $P(a,b \mid x,y)=\bra{\psi}E^a_x E^b_y\ket{\psi}$.
A behavior exhibits \emph{classical} correlations if there exists a unique joint probability distribution such that the behavior arises as marginals.
By a theorem of Fine \cite{fine1982hidden}, this also implies that classical behaviors are local.
We denote the sets of nonsignaling, quantum, and local behaviors by $N$, $Q$, and $L$, respectively.

Given a behavior, it is not easy to check whether there exists a corresponding quantum model (and thus whether the behavior belongs to $Q$).
\textcite{navascues2007bounding,navascues2008convergent}, in order to characterize the set of quantum behaviors, defined an infinite decreasing hierarchy of nonsignaling correlations (known as the NPA hierarchy).
These levels of correlations are intermediate; they are weaker than nonsignaling correlations, but stronger than the quantum set.
The different sets of behaviors in the NPA hierarchy are denoted by $Q_{1} \supseteq Q_{2} \supseteq \dots$, and converge to the quantum set in the sense that $\bigcap_{i \geq 1} Q_{i} = Q$.
Each set $Q_{i}$ can be certified by a different semidefinite program.

A further important set of supra-quantum behaviors are the \emph{almost quantum} correlations \cite{navascues2015almost}, which we denote $\tilde{Q} \supsetneq Q$.
It has been argued that this set is special, as it is the smallest set that contains the quantum set and arises naturally from some information theoretic principle (e.g. local orthogonality \cite{fritz2013local}, nontrivial communication complexity \cite{brassard2006limit}, etc.).
These correlations arise naturally by weakening a single one of the principles defining quantum correlations.
Namely, instead of requiring the local measurement operators to commute, one only requires that they commute when acting on the special state that gives the behavior, i.e. $[E_x^a,E_y^b]\ket{\psi}=0$.
It is shown in \cite{navascues2015almost} that $\tilde{Q} = Q_{1+AB}$, where $Q_{1+AB}$ is a set of correlations defined in \cite{navascues2008convergent} and satisfying $Q_{1} \supsetneq Q_{1+AB} \supsetneq Q_{2}$ in the NPA hierarchy.

Overall, the above correlations satisfy the inclusions
\begin{equation}
\label{eq:npa_inclusions}
    N \supsetneq Q_{1} \supsetneq Q_{1+AB} = \tilde{Q} \supsetneq Q_{2} \supseteq \dots \supseteq Q \supsetneq L .
\end{equation}
Here, it is worth stressing that the win probabilities in any game can only increase when considering a larger set of behaviors.
It follows that to (upper or lower) bound the win probabilities for players of a nonlocal game in one level, one can use other levels of correlations that are easier to deal with.
In this work, we will mainly be concerned with the nonsignaling, almost quantum, quantum, and local levels of correlations $N$, $\tilde{Q}$, $Q$, and $L$ respectively, where the almost quantum set is used to upper bound the win probabilities for quantum behaviors.

\section{Magic rectangle games: Definition}
\label{sec:magic_rectangle_games}

More generally than in \cref{sec:magic_square_game}, it is possible to construct similar games for arbitrary sizes of magic square; a magic square game with $m$ possible questions for Alice and $n$ for Bob corresponds to an \rectdim{m}{n} table.
Indeed, this may be more appropriately named a magic \emph{rectangle}.
In order to avoid trivially winning classical strategies, we must also generalize the game rules.
\begin{definition}[Magic rectangle games]
\label{def:magic_rectangle}
    We specify an \rectdim{m}{n} game by fixing some $\alpha_{1},\dots,\alpha_{m}$ and $\beta_{1},\dots,\beta_{n}$ each belonging to $\{+1, -1\}$, such that their product satisfies
    \begin{equation}
        \label{eq:prod_rule}
        \alpha_{1} \dots \alpha_{m} \cdot \beta_{1} \dots \beta_{n} = -1 .
    \end{equation}
    The rules of the given game are then:
    \begin{enumerate}[label=R\arabic*.,ref=R\arabic*]
        \item \label[rule]{rule:plus_minus} Each filled cell must belong to the set $\{+1, -1\}$.
        \item \label[rule]{rule:alice_prod} Upon being assigned the $i$th row, the product of Alice's entries must be $\alpha_{i}$.
        \item \label[rule]{rule:bob_prod} Upon being assigned the $j$th column, the product of Bob's entries must be $\beta_{j}$.
    \end{enumerate}
    As before, the game is won if both players enter the same value into their shared cell.
\end{definition}

Notice that the standard \rectdim{3}{3} magic square game described in \cref{sec:magic_square_game} is simply the special case where $\alpha_{1} = \alpha_{2} = \alpha_{3} = 1$ and $\beta_{1} = \beta_{2} = \beta_{3} = -1$.
In fact, there are $2^{m+n+1}$ different specifications of \rectdim{m}{n} games allowed by \cref{eq:prod_rule}.

The requirement of \cref{eq:prod_rule} ensures that no deterministic classical strategy that wins with certainty can exist.
In such a strategy, definite values would be assigned to each cell of the table which the players must both follow.
The product of all cells would be $\alpha_{1} \dots \alpha_{m}$ when calculated according to the rows, and $\beta_{1} \dots \beta_{n}$ according to the columns, but \cref{eq:prod_rule} is exactly the statement that these products are not equal.
Hence, the optimal classical success rate is at most $1 - (mn)^{-1}$.
In fact, this success rate is attainable deterministically by Alice and Bob answering according to fixed (but different) tables satisfying \cref{rule:plus_minus,rule:alice_prod,rule:bob_prod}, since such tables can always be constructed which differ in only a single one of their cells (Alice's table need not consider \cref{rule:bob_prod} and Bob's table need not consider \cref{rule:alice_prod}).
We denote this optimal classical success rate for our \rectdim{m}{n} magic rectangle games by
\begin{equation}
\label{eq:mxn_classical_value}
    \omega_{L}(m,n) = 1 - \frac{1}{mn} .
\end{equation}

Let us introduce some further notation to describe our magic rectangle games.
We will let $X$ and $Y$ be uniformly distributed random variables taking values in the alphabets $\mathcal{X} = \{1,\dots,m\}$ and $\mathcal{Y} = \{1,\dots,n\}$, respectively, labeling the possible input rows and columns that may be assigned to Alice and Bob.
We will denote the possible output rows of Alice and columns of Bob by the random vectors $\bm{A} = (A_{1},\dots,A_{n})$ and $\bm{B} = (B_{1},\dots,B_{m})^{\transpose}$ with alphabets $\mathcal{A}$ and $\mathcal{B}$, respectively, where each $A_{j}$ and $B_{i}$ takes values in $\{+1, -1\}$.
Referring to \cref{rule:plus_minus,rule:alice_prod,rule:bob_prod}, the event that the \rectdim{m}{n} magic rectangle game is won upon input $(X,Y) = (x,y)$ is given by
\begin{equation}
\label{eq:wincond}
    W_{x,y}^{m,n} \equiv (A_{y} = B_{x}) \cap \Bigg( \prod_{j=1}^{n} A_{j} = \alpha_{x} \Bigg) \cap \Bigg( \prod_{i=1}^{m} B_{i} = \beta_{y} \Bigg) .
\end{equation}

Perhaps more naturally for the games we consider, we can equivalently let $\mathscr{A}$ and $\mathscr{B}$ denote alphabets of the possible question/answer pairs for Alice and Bob allowed by the rules of \cref{def:magic_rectangle}.
To illustrate why this is the natural choice, we point out that Alice returning a string of $\pm 1$'s that is not compatible with \cref{rule:alice_prod} is equally forbidden with her returning the value $5$ for one cell, and thus it is the natural choice to exclude such outcomes from the alphabet altogether. This is mathematically expressed as
\begin{subequations}
\begin{align}
    \mathscr{A} & = \bigg\{ (x,\bm{a}) \in \mathcal{X}\times\mathcal{A} : \prod_{j} a_{j} = \alpha_{x} \bigg\} , \\
    \mathscr{B} & = \bigg\{ (y,\bm{b}) \in \mathcal{Y}\times\mathcal{B} : \prod_{i} b_{i} = \alpha_{y} \bigg\} .
\end{align}
\end{subequations}
Then, with $(X, \bm{A})$ and $(Y, \bm{B})$ instead taking values in alphabets $\mathscr{A}$ and $\mathscr{B}$, respectively, the winning event upon input $(X,Y) = (x,y)$ becomes simply
\begin{equation}
\label{eq:wincond_simple}
    A_{y} = B_{x} .
\end{equation}
We will refer to these $\mathscr{A}$ and $\mathscr{B}$ as the \emph{natural} alphabets of a magic rectangle game.

In what follows, we characterize the different sizes of magic rectangle games in terms of their optimal win probabilities and strategies, under different levels of allowed nonsignaling correlations (notably quantum, almost quantum, and general nonsignaling correlations).
We will often suppress the numerical values $+1$ and $-1$ to the symbols $+$ and $-$ for simplicity.

\section{Properties of magic rectangle games}
\label{sec:prelim_results}

To begin our characterization of the magic rectangle games of \cref{def:magic_rectangle}, we first show some general properties of these games, which allow us to narrow the considerations required for a full characterization.

\Cref{lem:value_unique} shows in what sense it is possible to identify games of the same dimension together.
\Cref{cor:value_unique} then shows that for magic rectangle games of a given dimension \rectdim{m}{n}, all choices of specific values for parameters $\alpha_{1},\dots,\alpha_{m}$ and $\beta_{1},\dots,\beta_{n}$ satisfying \cref{eq:prod_rule} yield the same optimal win probability at a given level of allowed correlations $\Sigma$.
We unambiguously refer to this value as $\omega_{\Sigma}(m,n)$ and show in \cref{cor:value_symmetric} the symmetry $\omega_{\Sigma}(m,n) = \omega_{\Sigma}(n,m)$.
We show in \cref{cor:value_increasing} that $\omega_{\Sigma}(m,n)$ is independently increasing in both $m$ and $n$ (with an explicit lower bound given in \cref{lem:value_lower_bound} in terms of that for smaller magic rectangle games).
Finally, the correlation hierarchy of \cref{eq:npa_inclusions} implies for any particular game
\begin{equation}
\label{eq:npa_winprob_ineq}
    \omega_{N} \geq \omega_{1} \geq \omega_{1+AB} \geq \omega_{2} \geq \dots \geq \omega_{Q} \geq \omega_{L} .
\end{equation}
Combining these facts leads us to the path we will take towards a characterization, as stated in the following theorem.
\begin{theorem}
\label{thm:chara_dims}
    In order to fully characterize quantum (or stronger) optimal strategies for magic rectangle games of arbitrary dimension, it is sufficient to consider only \rectdim{1}{n} games, \rectdim{2}{n} games with $n \geq 2$, and \rectdim{3}{3} games.
    Moreover, only a single example game for each different dimension need be considered.
\end{theorem}
\begin{proof}
    Postponed until the end of this section, after we have shown some general properties of magic rectangle games.
\end{proof}

\begin{definition}[Equivalence of games]
\label{def:game_equiv}
    We will call two games $G$ and $G^{\prime}$ \emph{equivalent}, and write $G \sim G^{\prime}$, if there exist bijections $f \colon \mathscr{A} \to \mathscr{A}^{\prime}$ and $g \colon \mathscr{B} \to \mathscr{B}^{\prime}$ taking the natural alphabets of $G$ to those of $G^{\prime}$, such that the winning events are equal. That is, such that $(X^{\prime},A^{\prime}) = f(X,A)$ and $(Y^{\prime},B^{\prime}) = g(Y,B)$ imply $W = W^{\prime}$, where $W$ and $W^{\prime}$ are the events that each game is won.
\end{definition}
\begin{remark}
    Under \cref{def:game_equiv}, given a fixed allowed level for correlations, all equivalent games have the same optimal win probability; strategies are identified with others of equal win probabilities.
\end{remark}

\begin{lemma}
\label{lem:binary_string}
    Let $b,b^{\prime} \in \{0,1\}^{n}$ be binary sequences of length $n \geq 2$ with the same parity (that is their Hamming weights are either both odd or both even).
    Consider the operations $\varphi_{i,j}$ on binary sequences, which have the effect of flipping the bits in both the $i$th and $j$th positions.
    Then, there exists an involutory composition of these operations $\varphi = \varphi_{i_{m},j_{m}} \circ\dots\circ \varphi_{i_{1},j_{1}}$ such that $b^{\prime} = \varphi(b)$.
\end{lemma}
\begin{proof}
    Starting with a binary sequence, we can apply operations $\varphi_{i,j}$ one-by-one in the following way:
    if there are two or more $1$'s in the sequence, apply the operation which replaces two of the $1$'s with $0$'s.
    If the initial binary sequence had even parity, repeating this process will eventually yield the sequence of zeros.
    Else, we will eventually have exactly one nonzero element in position $k$ of the sequence.
    If it is not already the case, we can apply $\varphi_{1,k}$ to take this to the sequence with exactly one nonzero element occurring in the first position.
    Hence, we can apply a sequence of these operations, taking each binary sequence to a canonical form depending only on its parity.
    Since each operation $\varphi_{i,j}$ is involutory, and the operations commute, any sequence of these operations is also involutory and thus invertible.
    Therefore we may apply some sequence of the operations $\varphi_{i_{m},j_{m}} \circ\dots\circ \varphi_{i_{1},j_{1}}$ taking $b$ to its canonical form, and from its canonical form to $b^{\prime}$.
\end{proof}

\begin{lemma}
\label{lem:value_unique}
    Let $G$ be an \rectdim{m}{n} magic rectangle game specified by the parameters $\alpha_{1},\dots,\alpha_{m}$ and $\beta_{1},\dots,\beta_{n}$ satisfying \cref{eq:prod_rule}, and let $G^{\prime}$ be a magic rectangle game of identical dimension specified by $\alpha_{1}^{\prime},\dots,\alpha_{m}^{\prime}$ and $\beta_{1}^{\prime},\dots,\beta_{n}^{\prime}$ also satisfying \cref{eq:prod_rule}.
    Then $G \sim G^{\prime}$ and, moreover, there exists an involution $F$ on the set of \rectdim{m}{n} games such that $G^{\prime} = F(G)$.
\end{lemma}
\begin{proof}
    Consider the operations $F_{i,j}$ which act on a game with parameters $\alpha_{1},\dots,\alpha_{m}$ and $\beta_{1},\dots,\beta_{n}$ to produce an identical game with exception that the sign of both $\alpha_{i}$ and $\beta_{j}$ have been flipped (this is a valid game as \cref{eq:prod_rule} is still satisfied).
    Correspondingly, let $f_{i,j}$ and $g_{i,j}$ act on the natural alphabets of the game to produce identical alphabets with the exceptions that each player changes the sign of their output corresponding to the $(i,j)$th cell of the table.
    That is, $f_{i,j}(X,A)$ differs from $(X,A)$ in that Alice flips the sign of $A_{j}$ if her input is $X=i$; similarly, in $g_{i,j}(Y,B)$, Bob flips the sign of $B_{i}$ if his input is $Y=j$.
    Upon applying $F_{i,j}$ to a game, the corresponding functions $f_{i,j}$ and $g_{i,j}$ leave the winning event \cref{eq:wincond_simple} unchanged for all possible inputs.
    Moreover, the $f_{i,j}$ and $g_{i,j}$ are bijective when considered as maps to the natural alphabets of the game produced by $F_{i,j}$.
    Hence, $F_{i,j}$ takes games to equivalent games.
    We will now show that we can apply some sequence of these operations $F = F_{i_{k},j_{k}} \circ\dots\circ F_{i_{1},j_{1}}$ such that $G^{\prime} = F(G)$.
    Transitivity of $\sim$ then shows the desired equivalence.

    Consider the parameters of $G$ as a binary sequence $b = (\alpha_{1},\dots,\alpha_{m}, \beta_{1},\dots,\beta_{n})$ containing an odd number of negative elements.
    The operation $F_{i,j}$ applied to $G$ acts to flip the sign of $\alpha_{i}$ and $\beta_{j}$.
    Furthermore, we can always construct an operation $F_{i_{2},j} \circ F_{i_{1},j}$ which flips the sign of $\alpha_{i_{1}}$ and $\alpha_{i_{2}}$, and similarly an operation $F_{i,j_{2}} \circ F_{i,j_{1}}$ which flips the sign of $\beta_{j_{1}}$ and $\beta_{j_{2}}$.
    Thus, by applying a sequence of these operations to $G$, we can flip the sign of any pair of its parameters in $b$.
    Therefore applying \cref{lem:binary_string} shows the existence of a sequence of these operations $F = F_{i_{k},j_{k}} \circ\dots\circ F_{i_{1},j_{1}}$ such that the game $F(G)$ has parameters given by the binary sequence (also containing an odd number of negative elements) $b^{\prime} = (\alpha_{1}^{\prime},\dots,\alpha_{m}^{\prime}, \beta_{1}^{\prime},\dots,\beta_{n}^{\prime})$.
    That is, $G^{\prime} = F(G)$.
    Finally, since the $F_{i,j}$ are involutory and commute with one another, $F$ is involutory.
\end{proof}

\begin{corollary}
\label{cor:value_unique}
    Given a fixed correlation level $\Sigma$, all magic rectangle games of dimension \rectdim{m}{n} have equal optimal win probability, which we denote $\omega_{\Sigma}(m,n)$.
\end{corollary}
\begin{proof}
    $G$ and $G^{\prime}$ in \cref{lem:value_unique} are arbitrary \rectdim{m}{n} games, and so all games of a fixed dimension are equivalent, and must have equal optimal win probabilities.
\end{proof}

\begin{definition}[Transpose game]
    We define the \emph{transpose} of an \rectdim{m}{n} game $G$ (with parameters $\alpha_{1},\dots,\alpha_{m}$ and $\beta_{1},\dots,\beta_{n}$), denoted by $G^{\transpose}$, to be the \rectdim{n}{m} game specified by the parameters $\alpha_{i}^{\transpose} = \beta_{i}$ and $\beta_{j}^{\transpose} = \alpha_{j}$ for all $i \in \{1,\dots,n\}$ and $j \in \{1,\dots,m\}$.
\end{definition}

\begin{lemma}
\label{lem:transpose_strat}
    Let $G$ be an \rectdim{m}{n} magic rectangle game, and fix an allowed level $\Sigma$ for correlations.
    If $S_{\Sigma}$ is a strategy for G which wins with probability $p$, then there exists an involution $T$ between strategies, such that the strategy $S_{\Sigma}^{\transpose} \equiv T(S_{\Sigma})$ for the transpose game $G^{\transpose}$ also wins with probability $p$.
\end{lemma}
\begin{proof}
    We let $T$ be the map which exchanges the roles of the players in a strategy, so that Bob's former strategy is now played by Alice, and vice versa.
    In particular, under the action of $T$, Alice in the transpose strategy $S_{\Sigma}^{\transpose}$ outputs Bob's columns of the strategy $S_{\Sigma}$ as rows.
    Similarly, Bob in $S_{\Sigma}^{\transpose}$ outputs Alice's rows of $S_{\Sigma}$ as columns.
    Such a $T$ is clearly involutory, and preserves the probability assigned to the winning event for magic rectangle games.
\end{proof}

\begin{lemma}
\label{lem:value_symmetric}
    Let $G$ be an \rectdim{m}{n} magic rectangle game, and let $G^{\prime}$ be an \rectdim{n}{m} magic rectangle game.
    Fix an allowed level $\Sigma$ for correlations.
    If $S_{\Sigma}$ is a strategy for $G$ which wins with probability $p$, then there exists a bijection $f$ between strategies such that the strategy $S_{\Sigma}^{\prime} = f(S_{\Sigma})$ for $G^{\prime}$ also wins with probability $p$.
\end{lemma}
\begin{proof}
    Let $S_{\Sigma}^{\transpose}$ be the \emph{transpose} strategy of $S_{\Sigma}$, obtained from \cref{lem:transpose_strat}.
    Then, $S_{\Sigma}^{\transpose}$ is a valid strategy for $G^{\transpose}$, which wins with probability $p$.
    By \cref{lem:value_unique}, $G^{\prime} \sim G^{\transpose}$, and so there exists a bijection $F$ such that the strategy $S_{\Sigma}^{\prime} = F(S_{\Sigma}^{\transpose})$ for $G^{\prime}$ also wins with probability $p$.
    The required function $f$ is defined by $f(S) = F(S^{\transpose})$.
\end{proof}

\begin{corollary}
\label{cor:value_symmetric}
    Optimal win probability is symmetric in the sense that
    \begin{equation}
        \omega_{\Sigma}(m,n) = \omega_{\Sigma}(n,m) .
    \end{equation}
\end{corollary}
\begin{proof}
    Let $S_{\Sigma}$ be an optimal strategy for an \rectdim{m}{n} game $G$, winning with probability $p$.
    Suppose that $S_{\Sigma}^{\prime}$ found from \cref{lem:value_symmetric} (also winning with probability $p$) is not optimal for an \rectdim{n}{m} game $G^{\prime}$.
    Then, there exists a strategy for $G^{\prime}$ which wins with probability $q > p$.
    Again by \cref{lem:value_symmetric}, this implies the existence of a strategy for $G$ which also wins with probability $q > p$, contradicting the optimality of $S_{\Sigma}$.
    Hence, $S_{\Sigma}^{\prime}$ is an optimal strategy for $G^{\prime}$.
    Since $G$ and $G^{\prime}$ were arbitrary, optimal strategies for all \rectdim{m}{n} and \rectdim{n}{m} games win with equal probability $p = \omega_{\Sigma}(m,n) = \omega_{\Sigma}(n,m)$.
\end{proof}

\begin{lemma}
\label{lem:value_lower_bound}
    Fix a level of allowed correlation $\Sigma$.
    Let the optimal win probability of \rectdim{m}{n} magic rectangle games be given by $\omega_{\Sigma}(m,n)$.
    If $m^{\prime} \geq m$ and $n^{\prime} \geq n$, then the optimal win probability of \rectdim{m^{\prime}}{n^{\prime}} games satisfies
    \begin{equation}
    \label{eq:value_lower_bound}
        \omega_{\Sigma}(m^{\prime},n^{\prime})
        \geq 1 - \frac{mn}{m^{\prime}n^{\prime}} [1 - \omega_{\Sigma}(m,n)] .
    \end{equation}
\end{lemma}
\begin{proof}
    Let $G$ be an \rectdim{m}{n} magic rectangle game specified by the parameters $\alpha_{1},\dots,\alpha_{m}$ and $\beta_{1},\dots,\beta_{n}$.
    From this, define an \rectdim{m^{\prime}}{n^{\prime}} game $G^{\prime}$ such that its parameters are
    \begin{subequations}
    \label{eq:value_lower_bound_params}
    \begin{align}
        \alpha_{i}^{\prime} & =
        \begin{cases}
            \alpha_{i} & \text{if $1 \leq i \leq m$,} \\
            1 & \text{if $m < i \leq m^{\prime}$,}
        \end{cases} \\
        \beta_{j}^{\prime} & =
        \begin{cases}
            \beta_{j} & \text{if $1 \leq j \leq n$,} \\
            1 & \text{if $n < j \leq n^{\prime}$.}
        \end{cases}
    \end{align}
    \end{subequations}
    Note that $G^{\prime}$ is indeed a valid game, as its parameters automatically satisfy \cref{eq:prod_rule}.
    Let $S_{\Sigma}$ be an optimal strategy for $G$, winning with probability $\omega_{\Sigma}(m,n)$, in which Alice outputs according to the random row vector $\bm{A} = (A_{1},\dots,A_{n})$ and Bob according to the random column vector $\bm{B} = (B_{1},\dots,B_{m})^{\transpose}$.
    Construct a strategy $S_{\Sigma}^{\prime}$ for $G^{\prime}$ in which Alice and Bob play their part of the strategy $S_{\Sigma}$ upon inputs $1 \leq X^{\prime} \leq m$ and $1 \leq Y^{\prime} \leq n$ respectively, but deterministically append $1$'s to their outputs to make up the required output length; upon other inputs, the players output only $1$'s.
    That is,
    \begin{subequations}
    \begin{align}
        \bm{A}^{\prime} & =
        \begin{cases}
            (A_{1},\dots,A_{n},1,\dots,1) & \text{if $1 \leq X^{\prime} \leq m$,} \\
            (1,\dots,1) & \text{if $m < X^{\prime} \leq m^{\prime}$,}
        \end{cases} \\
        \bm{B}^{\prime} & =
        \begin{cases}
            (B_{1},\dots,B_{m},1,\dots,1)^{\transpose} & \text{if $1 \leq Y^{\prime} \leq n$,} \\
            (1,\dots,1)^{\transpose} & \text{if $n < Y^{\prime} \leq n^{\prime}$.}
        \end{cases}
    \end{align}
    \end{subequations}
    It is clear that these outputs always satisfy the rules given in \cref{def:magic_rectangle} for the parameters of $G^{\prime}$ defined in \cref{eq:value_lower_bound_params}.
    Moreover, by using strategy $S_{\Sigma}^{\prime}$, the players succeed at $G^{\prime}$ with probability $\omega_{\Sigma}(m,n)$ upon $mn$ of the $m^{\prime}n^{\prime}$ possible inputs, and with certainty upon the remaining inputs.
    By \cref{cor:value_unique}, the win probability of $S_{\Sigma}^{\prime}$ at the \rectdim{m^{\prime}}{n^{\prime}} game $G^{\prime}$ is at most the optimal win probability for \rectdim{m^{\prime}}{n^{\prime}} games $\omega_{\Sigma}(m^{\prime},n^{\prime})$.
    Hence, since the inputs are chosen uniformly at random,
    \begin{equation}
        \omega_{\Sigma}(m^{\prime},n^{\prime})
        \geq \frac{mn}{m^{\prime}n^{\prime}} \omega_{\Sigma}(m,n) + \frac{m^{\prime}n^{\prime} - mn}{m^{\prime}n^{\prime}} ,
    \end{equation}
    which is exactly \cref{eq:value_lower_bound}.
\end{proof}

\begin{corollary}
\label{cor:value_increasing}
    Fix a correlation level $\Sigma$, and let $m^{\prime} \geq m$ and $n^{\prime} \geq n$. Then
    \begin{equation}
        \omega_{\Sigma}(m^{\prime},n^{\prime}) \geq \omega_{\Sigma}(m,n) .
    \end{equation}
\end{corollary}
\begin{proof}
    Immediate from \cref{eq:value_lower_bound} upon noting $\frac{mn}{m^{\prime} n^{\prime}} \leq 1$ and $\omega_{\Sigma}(m,n) \leq 1$.
\end{proof}

Having stated and proven the preceding properties of magic rectangle games, it is now easy to see that \cref{thm:chara_dims} holds as follows.

\begin{proof}[Proof of \cref{thm:chara_dims}]
    The second part of the claim (that only a single example game for each different dimension need be considered) is shown by \cref{lem:value_unique} and \cref{cor:value_unique}, which state that all games of the same dimension are equivalent.

    For the first part of the claim, we may first choose to consider optimal strategies for \rectdim{1}{n} games.
    Then, by \cref{lem:value_symmetric} and \cref{cor:value_symmetric}, there are invertible maps between optimal strategies for \rectdim{n}{1} games and \rectdim{1}{n} games.
    Thus we next study \rectdim{2}{n} games without the need to consider the \rectdim{2}{1} case.
    Similarly, we then need not consider \rectdim{n}{2} cases.
    Finally, considering the following observations, we can see that all \rectdim{m}{n} games where both $m \geq 3$ and $n \geq 3$ can be won with certainty for quantum (or stronger) behaviors.
    It was pointed out in \cref{sec:magic_square_game} that quantum strategies for the standard \rectdim{3}{3} magic square game which win with certainty are already known.
    As the rules \cref{rule:stand_plus_minus,rule:stand_alice_prod,rule:stand_bob_prod} for the standard \rectdim{3}{3} magic square game are a special case of our magic rectangle games given in \cref{def:magic_rectangle}, the existence of quantum winning strategies for all general \rectdim{3}{3} games is guaranteed by \cref{lem:value_unique}.
    Therefore, since by \cref{cor:value_increasing} the quantum value $\omega_{Q}(m,n)$ is increasing in $m$ and $n$, and noting the inequalities of \cref{eq:npa_winprob_ineq}, all magic rectangle games with $m \geq 3$ and $n \geq 3$ satisfy $\omega_{\Sigma}(m,n) = 1$, where $\Sigma$ is any nonsignaling correlation level at least as strong as the quantum set.
    Furthermore, the proof of \cref{lem:value_lower_bound} combined with \cref{lem:value_unique} shows how to construct winning strategies for all such games from a winning \rectdim{3}{3} strategy.
    Hence, the \rectdim{3}{3} games already studied are the final case required to complete the characterization of magic rectangle games.
\end{proof}

\section{Characterization of magic rectangles}
\label{sec:characterization}

Following \cref{thm:chara_dims}, we characterize magic rectangle games of all sizes by considering those of dimension \rectdim{1}{n} for $n \geq 1$ and \rectdim{2}{n} for $n \geq 2$.
The final \rectdim{3}{3} case was already discussed in \cref{sec:magic_square_game}.

\begin{theorem}
\label{thm:characterization}
    The optimal success probabilities of all magic rectangle games can be characterized as follows:
    \begin{enumerate}
        \item
        \label{item:1xn}
            Games of dimension \rectdim{1}{n} cannot exhibit superclassical behavior;
            \begin{equation}
                \omega_{N}(1,n) = \omega_{L}(1,n) = 1 - \frac{1}{n} .
            \end{equation}
        \item
        \label{item:2xn}
            Games of dimension \rectdim{2}{n} for $n \geq 2$ satisfy
            \begin{equation}
                1 - \frac{2 - \sqrt{2}}{2n} \leq \omega_{Q}(2,n) \leq \omega_{1+AB}(2,n) = \frac{1}{2} {\left(1 + \sqrt{1 - \frac{1}{n}} \right)} ,
            \end{equation}
            where the final equality is conjectured, with strong numerical evidence for $n \leq 6$.
            Such games can be won with certainty in the general nonsignaling regime;
            \begin{equation}
                \omega_{N}(2,n) = 1 .
            \end{equation}
            Moreover, for NPA hierarchy level 1 (or stronger) correlations and $n \geq 3$,
            \begin{equation}
                \omega_{1}(2,n) = 1 .
            \end{equation}
        \item
        \label{item:3x3_and_larger}
            For all quantum or stronger correlations, games of dimension \rectdim{m}{n} where both $m \geq 3$ and $n \geq 3$ can be won with certainty;
            \begin{equation}
                \omega_{Q}(m,n) = 1 .
            \end{equation}
    \end{enumerate}
\end{theorem}
\begin{proof}
    The content of \cref{item:1xn} is \cref{thm:1xn_is_classical}.
    The discussion in \cref{sec:2xn_games} covers \cref{item:2xn}.
    \Cref{item:3x3_and_larger} was discussed as part of the proof of \cref{thm:chara_dims}, and can be seen by combining \cref{cor:value_increasing} with the fact that $\omega_{Q}(3,3) = 1$ by \cref{cor:value_unique}.
\end{proof}

\subsection{1-by-n magic rectangles}
\label{sec:1xn_games}

\begin{theorem}
\label{thm:1xn_is_classical}
    Under any set of nonsignaling correlations, the optimal win probability of \rectdim{1}{n} games coincides with the classical value,
    \begin{equation}
        \omega_{N}(1,n) = \omega_{L}(1,n) = 1 - \frac{1}{n} .
    \end{equation}
\end{theorem}
\begin{proof}
    For all possible inputs $Y=j$ for Bob, his single output value is deterministically equal to $\beta_{j}$ according to \cref{rule:bob_prod} of \cref{def:magic_rectangle}.
    However, recalling \cref{eq:prod_rule} and denoting the product of Alice's single output row by $\alpha$, we require any valid \rectdim{1}{n} game to satisfy $\alpha \neq \beta_{1} \dots \beta_{n}$.
    That is, Alice's output row must contain at least one element, in position $k$ say, which differs from the output value $\beta_{k}$ Bob would give if his input was $Y=k$.
    By the assumption of no-signaling, Alice cannot have any knowledge about which of $n$ possible uniform inputs was provided to Bob.
    Thus the probability of the losing event that $A_{k} \neq \beta_{k}$ (the element of Alice's output corresponding to Bob's input differs from Bob's output) is at least $n^{-1}$.
    Therefore $\omega_{N}(1,n) \leq 1 - n^{-1} = \omega_{L}(1,n)$.
    Since trivially also $\omega_{N}(1,n) \geq \omega_{L}(1,n)$ by \cref{eq:npa_winprob_ineq}, we have the result.
\end{proof}

\subsection{2-by-n magic rectangles}
\label{sec:2xn_games}

Before discussing the general case of \rectdim{2}{n} magic rectangle games, let us first examine the special case of \rectdim{2}{2} magic square games.

\subsubsection{2-by-2 magic squares}
\label{sec:2x2_magic_squares}

In this case, \cref{eq:prod_rule} states that either exactly one of the possible rows or columns is required to have a negative product, or exactly one is required to have a positive product.
In fact, any such \rectdim{2}{2} magic square game can be identified with the well-known CHSH game, in which Alice and Bob are provided binary inputs $X_{\text{CHSH}} \in \{0,1\}$ and $Y_{\text{CHSH}} \in \{0,1\}$ uniformly at random, and win by returning binary outputs $A_{\text{CHSH}} \in \{0,1\}$ and $B_{\text{CHSH}} \in \{0,1\}$ which satisfy \cite{clauser1969proposed}
\begin{equation}
\label{eq:wincond_chsh}
    A_{\text{CHSH}} \oplus B_{\text{CHSH}} = X_{\text{CHSH}} \land Y_{\text{CHSH}} .
\end{equation}

We will now explicitly construct this equivalence, whereupon we note the statement $\omega_{L}(2,2) = \frac{3}{4}$ defines the unique nontrivial facet of the local polytope in the $(2,2,2)$ Bell scenario (which corresponds also to the CHSH inequality) \cite{fine1982hidden,pironio2004aspects}.

\begin{theorem}
\label{thm:2x2_chsh_equiv}
    Any \rectdim{2}{2} magic square game is equivalent (in the sense of \cref{def:game_equiv}) to the CHSH game.
\end{theorem}
\begin{proof}
    Consider the \rectdim{2}{2} magic square with specified row products $(\alpha_{1}, \alpha_{2}) = (+,+)$ and column products $(\beta_{1}, \beta_{2}) = (+,-)$.
    We first show that this game is equivalent to the CHSH game.
    Then, since all \rectdim{2}{2} games are equivalent (\cref{lem:value_unique}), the desired result follows by transitivity.

    We can identify the input events of the two games as
    \begin{subequations}
    \begin{align}
        X_{\text{CHSH}} = 0 &\longleftrightarrow X = 1 , \\
        X_{\text{CHSH}} = 1 &\longleftrightarrow X = 2
    \end{align}
    \end{subequations}
    for Alice, and for Bob
    \begin{subequations}
    \begin{align}
        Y_{\text{CHSH}} = 0 &\longleftrightarrow Y = 1 , \\
        Y_{\text{CHSH}} = 1 &\longleftrightarrow Y = 2 .
    \end{align}
    \end{subequations}
    Alice identifies her two possible outputs as simply
    \begin{subequations}
    \begin{align}
        A_{\text{CHSH}} = 0 &\longleftrightarrow \bm{A} = (+,+) , \\
        A_{\text{CHSH}} = 1 &\longleftrightarrow \bm{A} = (-,-) .
    \end{align}
    \end{subequations}
    Bob identifies his outputs depending on his assigned input.
    If $Y_{\text{CHSH}} = 0$ (equivalently $Y = 1$), then he makes the identifications
    \begin{subequations}
    \begin{align}
        B_{\text{CHSH}} = 0 &\longleftrightarrow \bm{B} = (+,+)^{\transpose} , \\
        B_{\text{CHSH}} = 1 &\longleftrightarrow \bm{B} = (-,-)^{\transpose} .
    \end{align}
    \end{subequations}
    However, if $Y_{\text{CHSH}} = 1$ (equivalently $Y = 2$), then he makes alternative identifications
    \begin{subequations}
    \begin{align}
        B_{\text{CHSH}} = 0 &\longleftrightarrow \bm{B} = (+,-)^{\transpose} , \\
        B_{\text{CHSH}} = 1 &\longleftrightarrow \bm{B} = (-,+)^{\transpose} .
    \end{align}
    \end{subequations}
    These identifications form bijections $f \colon \mathscr{A}_{\text{CHSH}} \to \mathscr{A}$ and $g \colon \mathscr{B}_{\text{CHSH}} \to \mathscr{B}$ between the natural alphabets of each game, and are explicitly tabulated in \cref{tab:2x2_chsh_biject}.
    \begin{table}[htb]
        \centering
        \small
        \caption{
            The bijections $f \colon \mathscr{A}_{\text{CHSH}} \to \mathscr{A}$ and $g \colon \mathscr{B}_{\text{CHSH}} \to \mathscr{B}$ used to show the equivalence between the CHSH game and the \rectdim{2}{2} magic square game with parameters $(\alpha_{1},\alpha_{2}) = (+,+)$ and $(\beta_{1},\beta_{2}) = (+,-)$.
            Elements of the natural alphabets $\mathscr{A}$, $\mathscr{B}$, $\mathscr{A}_{\text{CHSH}}$, and $\mathscr{B}_{\text{CHSH}}$ have the form of possible input/output pairs for each game and player, with the input written first.
            }
        \label{tab:2x2_chsh_biject}
        \begin{tabular}{cccc}
            \toprule
            \multicolumn{2}{c}{$f$} & \multicolumn{2}{c}{$g$} \\
            \cmidrule(r){1-2}\cmidrule(l){3-4}
            $\mathscr{A}_{\text{CHSH}}$ & $\mathscr{A}$ & $\mathscr{B}_{\text{CHSH}}$ & $\mathscr{B}$ \\
            \midrule
            $(0,0)$ & $(1, (+,+))$ & $(0,0)$ & $(1, (+,+)^{\transpose})$ \\
            $(0,1)$ & $(1, (-,-))$ & $(0,1)$ & $(1, (-,-)^{\transpose})$ \\
            $(1,0)$ & $(2, (+,+))$ & $(1,0)$ & $(2, (+,-)^{\transpose})$ \\
            $(1,1)$ & $(2, (-,-))$ & $(1,1)$ & $(2, (-,+)^{\transpose})$ \\
            \bottomrule
        \end{tabular}
    \end{table}

    It remains to show that the winning event for the CHSH game, \cref{eq:wincond_chsh}, and the winning event for the \rectdim{2}{2} magic rectangle game of \cref{eq:wincond_simple} upon any input,
    \begin{equation}
    \label{eq:wincond_anyinput}
        \bigcup_{\mathclap{x,y \in \{1,2\}}} \; [(A_{y} = B_{x}) \cap (X=x) \cap (Y=y)] ,
    \end{equation}
    are identical under the functions $f$ and $g$.
    We can rewrite these two events to more closely resemble one another as
    \begin{equation}
        \bigcup_{\mathclap{x,y \in \{0,1\}}} \; [(A_{\text{CHSH}} \oplus B_{\text{CHSH}} = x \land y) \cap (X_{\text{CHSH}} = x) \cap (Y_{\text{CHSH}} = y)]
    \end{equation}
    for \cref{eq:wincond_chsh}, and for \cref{eq:wincond_anyinput}
    \begin{equation}
        \bigcup_{\mathclap{x,y \in \{0,1\}}} \; [(A_{y+1} = B_{x+1}) \cap (X=x+1) \cap (Y=y+1)] .
    \end{equation}
    One can verify from the identifications made (for example by examining \cref{tab:2x2_chsh_biject}) that terms in the first union above are pairwise equal to those in the second.
    That is, for all $x,y \in \{0,1\}$,
    \begin{equation}
    \begin{split}
        [(& A_{\text{CHSH}} \oplus B_{\text{CHSH}} = x \land y) \cap (X_{\text{CHSH}} = x) \cap (Y_{\text{CHSH}} = y)] \\
        &\equiv [(A_{y+1} = B_{x+1}) \cap (X=x+1) \cap (Y=y+1)] .
    \end{split}
    \end{equation}
\end{proof}

\begin{corollary}
\label{cor:2x2}
The maximum probability with which the \rectdim{2}{2} magic square game can be won is (i) $(2 + \sqrt{2}) / 4 \approx 0.854$ for quantum strategies and (ii) unity for general nonsignaling strategies.
\end{corollary}
\begin{proof}
    The result of \cref{thm:2x2_chsh_equiv} means that the maximum attainable win probability for any quantum strategy coincides with that of the CHSH game, namely $(2 + \sqrt{2}) / 4 \approx 0.854$.
    For the same reason, under PR box assumptions \cite{popescu1998causality}, the \rectdim{2}{2} magic square game can be won with certainty.
\end{proof}

An example of the identifications made for the \rectdim{2}{2} magic square game considered in the proof of \cref{thm:2x2_chsh_equiv} is depicted in \cref{fig:2x2_chsh_example}.
\begin{figure}[hbt]
    \centering
    \includegraphics{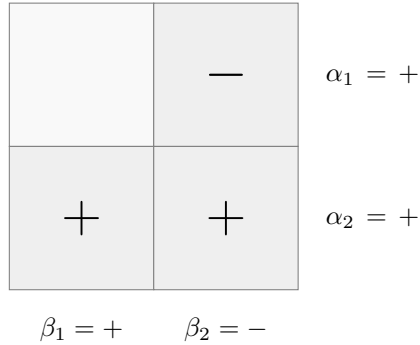}
    \caption{
        Example of the equivalence of the \rectdim{2}{2} magic square and CHSH games.
        Shown is a filled \rectdim{2}{2} magic square with row products $(\alpha_{1}, \alpha_{2}) = (+,+)$ and column products $(\beta_{1}, \beta_{2}) = (+,-)$ specified.
        The input row and column $X = 2$ and $Y = 2$ were chosen for this example.
        Alice gave output $\bm{A} = (+,+)$ and Bob gave output $\bm{B} = (-,+)^{\transpose}$.
        The game is won since $A_{2} = B_{2}$.
        The equivalent input and output configuration for the CHSH game, using the identifications of \cref{tab:2x2_chsh_biject}, are $(X_{\text{CHSH}},A_{\text{CHSH}}) = (1,0)$ and $(Y_{\text{CHSH}},B_{\text{CHSH}}) = (1,1)$.
        The CHSH win condition of \cref{eq:wincond_chsh} is also satisfied.
        }
    \label{fig:2x2_chsh_example}
\end{figure}

\subsubsection{General 2-by-n games}
\label{sec:general_2xn_games}

As stated in \cref{thm:chara_dims}, it is enough to consider $n \geq 2$.
From \cref{eq:mxn_classical_value}, the optimal classical win probability for \rectdim{2}{n} games is given by
\begin{equation}
\label{eq:2xn_classical_value}
    \omega_{L}(2,n) = 1 - \frac{1}{2n} .
\end{equation}
Using the discussion of \cref{sec:2x2_magic_squares}, we can apply \cref{lem:value_lower_bound} to an optimal \rectdim{2}{2} quantum strategy with value $\omega_{Q}(2,2) = (2 + \sqrt{2}) / 4$ as given by \cref{cor:2x2}.
The win probability of the resulting \rectdim{2}{n} strategy lower bounds the \rectdim{2}{n} quantum value via \cref{eq:value_lower_bound} as
\begin{equation}
\label{eq:2xn_lower_bound}
    \omega_{Q}(2,n) \geq 1 - \frac{2 - \sqrt{2}}{2n} .
\end{equation}
In order to find an upper bound for this quantum value, we have used the implementation of the NPA hierarchy found in the \textsc{Ncpol2sdpa} \cite{wittek2015algorithm} package with the \textsc{MOSEK} \cite{mosek2020pythonapi} semidefinite program solver.
Optimal values for different \rectdim{2}{n} games and levels of the hierarchy are shown in \cref{tab:npa_values}.
\begin{table}[htb]
    \centering
    \small
    \caption{
        Optimal win probabilities for \rectdim{2}{n} magic rectangle games, under correlations allowed by different levels of the NPA hierarchy.
        We see that, for the cases tested, the optimal win probabilities are identical at every level beyond the almost quantum $1+AB$ level.
        Moreover, these values appear to follow exactly the expression given in \cref{eq:2xn_upper_bound}.
        For $n \geq 3$, we observe games which can be won with certainty at level 1, but with lower than unit probability at the almost quantum and higher levels.
        Values were obtained through \textsc{Ncpol2sdpa} with the \textsc{MOSEK} solver.
        Results were also verified with the \textsc{QETLAB} \cite{qetlab} toolbox, using \textsc{MOSEK} \cite{mosek2020toolbox} within \textsc{CVX} \cite{cvx}.
        }
    \label{tab:npa_values}
    \begin{tabular}{cccccc}
        \toprule
        & \multicolumn{5}{c}{NPA hierarchy level} \\
        \cmidrule(l){2-6}
        $n$ & 1 & $1+AB$ & 2 & 3 & 4 \\
        \midrule
        2 & 0.8535533906 & 0.8535533906 & 0.8535533906 & 0.8535533906 & 0.8535533906 \\
        3 & 1.0000000000 & 0.9082482905 & 0.9082482905 & 0.9082482905 & 0.9082482905 \\
        4 & 1.0000000000 & 0.9330127019 & 0.9330127019 & & \\
        5 & 1.0000000000 & 0.9472135955 & 0.9472135955 & & \\
        6 & 1.0000000000 & 0.9564354646 & & & \\
        \bottomrule
    \end{tabular}
\end{table}

We note that for all levels $1+AB$ and above that were tested, the optimal value is identical for each \rectdim{2}{n} game, and appears to bound above the quantum value for $n \leq 6$ by the closed-form expression
\begin{equation}
\label{eq:2xn_upper_bound}
    \omega_{Q}(2,n) \leq \omega_{1+AB}(2,n) = \frac{1}{2} {\left(1 + \sqrt{1 - \frac{1}{n}} \right)} .
\end{equation}
Furthermore, since the complete bipartite graph $K_{2,n}$ is planar for all $n$, we know from \cite[Theorem~21]{arkhipov2012extending} that $\omega_{Q}(2,n) < 1$.
The classical value given by \cref{eq:2xn_classical_value} and the quantum bounds given by \cref{eq:2xn_lower_bound,eq:2xn_upper_bound} are depicted in \cref{fig:2xn_bounds}.
\begin{figure}[hbt]
    \centering
    \includegraphics{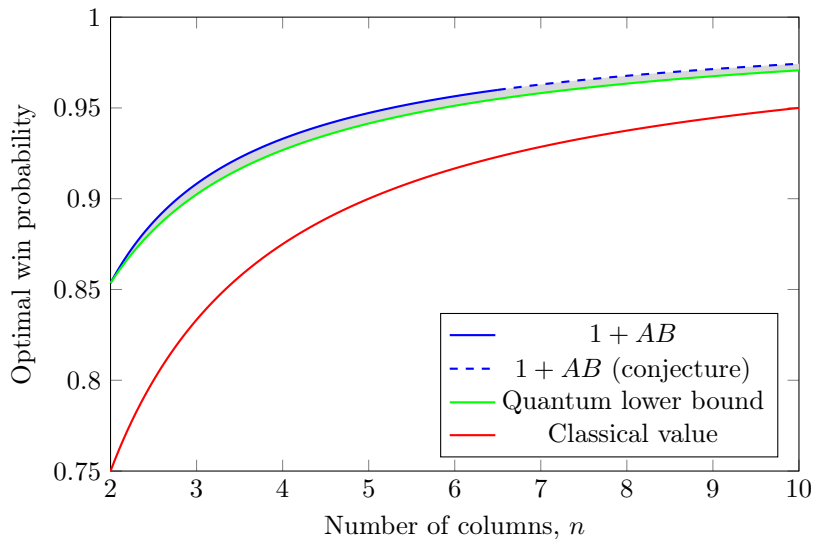}
    \caption{
        Bounds on the optimal quantum win probability of \rectdim{2}{n} magic rectangle games.
        The lowermost curve is the classical value for each game, given by \cref{eq:2xn_classical_value}.
        The middle curve is the lower bound of \cref{eq:2xn_lower_bound} on the quantum value of each game, resulting from application of \cref{lem:value_lower_bound} to the optimal quantum value for \rectdim{2}{2} games.
        The solid upper curve shows the maximal almost quantum win probability (see NPA hierarchy level $1+AB$ of \cref{tab:npa_values}), which provides an upper bound to the quantum value; where the line is dashed corresponds to our conjectured values for large $n$, given by \cref{eq:2xn_upper_bound}, which have proved to be too computationally intensive to test.
        The region within which the quantum values could possibly lie is shaded.
        }
    \label{fig:2xn_bounds}
\end{figure}

\begin{conjecture}
\label{conj:2xn_upper_bound}
    The expression for $\omega_{1+AB}(2,n)$ given in \cref{eq:2xn_upper_bound} holds for all $n \geq 1$.
\end{conjecture}
\begin{remark}
    Using the \textsc{SDPA-GMP} \cite{sdpagmp,yamashita2010high,nakata2010numerical} semidefinite program solver with arbitrary-precision arithmetic, we have been able to verify agreement of \cref{eq:2xn_upper_bound} with all but the most computationally intensive entries of \cref{tab:npa_values} to a much higher precision than printed.
\end{remark}

Since under general no-signaling assumptions the \rectdim{2}{2} magic square game can be won with certainty (\cref{cor:2x2}), so too can all \rectdim{2}{n} games with $n \geq 2$ by \cref{cor:value_increasing}.
It is interesting to note that, as far as the authors are aware, those \rectdim{2}{n} games for $n \geq 3$ examined in \cref{tab:npa_values} are the first examples of nonlocal games with the property that they can be won with certainty using NPA hierarchy level 1 correlations, but only with less than unit probability using almost quantum level $1+AB$ correlations.
An explicit strategy for winning the \rectdim{2}{3} game with certainty using NPA hierarchy level 1 correlations is given in \cref{sec:2x3_at_npa1}.
Hence, by \cref{cor:value_increasing}, the result that $\omega_{1}(2,n) = 1$ for all $n \geq 3$ is exact.

\section{Application to certified randomness expansion}
\label{sec:app_randomness}

In this section, we will be concerned with utilizing the Bell inequality violations provided by magic rectangle games to achieve certified randomness expansion, using the device-independent spot-checking protocol $R_{gen}$ described in \cite[Figure~2]{miller2017universal}.
The main technical result of this section is to relate the win probabilities of \rectdim{m}{n} magic rectangle games with distinguished input, to those of \rectdim{(m-1)}{(n-1)} games.
This enables us to get the optimal noise tolerance of such games, as well as to simply obtain rates for randomness expansion using general magic rectangle games.
In terms of rates, there are new techniques that could improve our results, but would need to be examined on a case-by-case basis (see also \cref{sec:randomness_discussion}).

Given a nonlocal game, we will denote by $\omega$ its optimal win probability over quantum devices, and by $\bar{\omega}$ its optimal win probability over quantum devices with a \emph{distinguished input} (that is, devices which give deterministic outputs upon a single distinguished choice of input).
Protocol $R_{gen}$ is shown to produce quantum-secure extractable bits over $N$ rounds, provided its \emph{score acceptance threshold} parameter satisfies $\chi > \bar{\omega}$.
In our notation, this result can be stated as
\begin{theorem}[{\cite[Theorem~1.1]{miller2017universal}}]
\label{thm:miller_noise}
    For any game, there are functions $\pi \colon [0,\omega] \to \mathbb{R}_{\geq 0}$ and $\Delta \colon (0,1]^{2} \to \mathbb{R}_{\geq 0}$ such that the following hold:
    \begin{enumerate}
        \item For any $b \in (0,1]$, Protocol~$R_{gen}$ produces at least $N[\pi(\chi) - \Delta(b,q)]$ extractable bits with soundness error $3 \cdot 2^{-bqN}$.
        \item The function $\pi$ is nonzero on the interval $(\bar{\omega}, \omega]$.
        \item The function $\Delta$ tends to $0$ as $(b,q) \to (0,0)$.
    \end{enumerate}
\end{theorem}
Modeling noise as a process in which an adversary is allowed to change the outputs of a device arbitrarily with some probability, the noise tolerance of the protocol is $\omega - \chi$ (the adversary is allowed to change the expected score at the game by at most this amount).
The noise tolerance is then maximally $\omega - \bar{\omega}$.

Furthermore, an explicit lower bound on the function $\pi$ was proved in \cite{miller2017universal}, and can be stated as follows.
\begin{theorem}[{\cite[Theorem~5.8]{miller2017universal}}]
\label{thm:miller_rate}
    Let $G$ be a game with output alphabet size $r \geq 2$, and let $\bar{\omega}$ be the maximum win probability of this game over compatible devices with a distinguished input.
    Then, the following function is a rate curve:
    \begin{equation}
        \pi(\chi) = \begin{cases}
            \frac{2(\log_{2}{e})(\chi - \bar{\omega})^{2}}{r-1} & \text{if $\chi > \bar{\omega}$,} \\
            \hfil 0 & \text{otherwise.}
        \end{cases}
    \end{equation}
\end{theorem}

\subsection{Win probability with distinguished input}
\label{sec:distinguished_bound}

Since \rectdim{1}{n} magic rectangle games do not exhibit superclassical behavior (\cref{thm:1xn_is_classical}), such games cannot be used in randomness expansion.
We construct an optimal strategy for arbitrary \rectdim{m}{n} magic rectangle games having a distinguished input, where $m,n \geq 2$.

\begin{theorem}
\label{thm:distinguished_bound}
    Fix an allowed level $\Sigma$ for nonsignaling correlations.
    The optimal win probability for any \rectdim{m}{n} magic rectangle game having a distinguished input, with $m \geq 2$ and $n \geq 2$, is given by
    \begin{equation}
    \label{eq:distinguished_bound}
        \bar{\omega}_{\Sigma}(m,n) =
        1 - \frac{(m-1)(n-1)}{mn}[1 - \omega_{\Sigma}(m-1,n-1)] .
    \end{equation}
    A strategy which attains this value is to play an optimal strategy for \rectdim{(m-1)}{(n-1)} games, but with all output strings extended to include one deterministic entry.
\end{theorem}
\begin{proof}
Without loss of generality, let us choose this distinguished input to be given by the event $(X = 1) \cap (Y = 1)$.
Recall that the event that the game is won upon some input is given in \cref{eq:wincond}.
We will let $W_{x,y} \equiv W_{x,y}^{m,n}$ throughout the following for brevity.

By imposing the no-signaling principle, we see that for all inputs $x \in \{1,\dots,m\}$ and $y \in \{1,\dots,n\}$, there exists an output entry $a^{x} \in \{+1,-1\}$ for Alice such that
\begin{equation}
\label{eq:determined_A}
\begin{split}
    P( & A_{1} = a^{x} \mid W_{X,Y} \cap X = x \cap Y = y) \\
    &= P(A_{1} = a^{x} \mid W_{X,Y} \cap X = x \cap Y = 1) \\
    &= P(B_{x} = a^{x} \mid W_{X,Y} \cap X = x \cap Y = 1) \\
    &= P(B_{x} = a^{x} \mid W_{X,Y} \cap X = 1 \cap Y = 1) \\
    &= P(B_{x} = a^{x} \mid W_{X,Y} \cap Y = 1) = 1 ,
\end{split}
\end{equation}
where the second equality uses our conditioning on \cref{eq:wincond}; the first, third, and fourth equalities use no-signaling; and the final equality comes from our choice of distinguished input.
Similarly, there exists an output $b^{y}$ for Bob such that
\begin{equation}
\label{eq:determined_B}
\begin{split}
    P( & B_{1} = b^{y} \mid W_{X,Y} \cap X = x \cap Y = y) \\
    &= P(A_{y} = b^{y} \mid W_{X,Y} \cap X = 1) = 1 .
\end{split}
\end{equation}
Combining \cref{eq:determined_A,eq:determined_B} yields
\begin{equation}
\label{eq:determined}
    P(A_{1} = a^{x} \cap B_{1} = b^{y} \mid W_{X,Y} \cap X = x \cap Y = y) = 1 .
\end{equation}
Now, since for arbitrary events $W$, $E$, and $F$ we have
\begin{equation}
    P(E \mid W \cap F) = 1 \implies P(W \mid F) = P(W \cap E \mid F) ,
\end{equation}
from \cref{eq:determined} we can see
\begin{equation}
\label{eq:expanded_term}
    P(W_{x,y} \mid X = x \cap Y = y) = P(W_{x,y} \cap A_{1} = a^{x} \cap B_{1} = b^{y} \mid X = x \cap Y = y) .
\end{equation}

We can now calculate the win probability for a device with a distinguished input. Expanding according to the uniformly distributed input variables and applying the result of \cref{eq:expanded_term} gives
\begin{equation}
\begin{split}
\label{eq:winprob}
    P(W_{X,Y})
    &= \frac{1}{mn} \sum_{x, y} P(W_{x,y} \mid X = x \cap Y = y) \\
    &= \frac{1}{mn} \sum_{x, y} P(W_{x,y} \cap A_{1} = a^{x} \cap B_{1} = b^{y} \mid X = x \cap Y = y) .
\end{split}
\end{equation}
It is clear that if $a^{1} \neq b^{1}$ then $W_{1,1} = \varnothing$, and the first term of \cref{eq:winprob} vanishes so that $P(W_{X,Y}) \leq 1 - (mn)^{-1}$.
Let us now assume that $a^{1} = b^{1}$.
In the case where $\prod_{j=1}^{n} b^{j} \neq \alpha_{1}$, we can bound the terms of \cref{eq:winprob} where $X=1$ as
\begin{equation}
\begin{split}
    & \sum_{y=1}^{n} P(W_{1,y} \cap A_{1} = a^{1} \cap B_{1} = b^{y} \mid X = 1 \cap Y = y) \\
    {}\leq{} & \sum_{y=1}^{n} P(A_{y} = b^{y} \cap \textstyle\prod_{j=1}^{n} A_{j} = \alpha_{1} \mid X = 1)
    \leq n - 1 .
\end{split}
\end{equation}
Similarly, in the case where $\prod_{i=1}^{m} a^{i} \neq \beta_{1}$, we can bound the terms where $Y=1$ as
\begin{equation}
\begin{split}
    & \sum_{x=1}^{m} P(W_{x,1} \cap A_{1} = a^{x} \cap B_{1} = b^{1} \mid X = x \cap Y = 1) \\
    {}\leq{} & \sum_{x=1}^{m} P(B_{x} = a^{x} \cap \textstyle\prod_{i=1}^{m} B_{i} = \beta_{1} \mid Y = 1)
    \leq m - 1 .
\end{split}
\end{equation}
Therefore, we have shown $P(W_{X,Y}) \leq 1 - (mn)^{-1} = \omega_{L}(m,n)$ in all cases other than where
\begin{equation}
\label{eq:remaining_cases}
    \left( a^{1} = b^{1} \right) \land \Bigg( \prod_{i=1}^{m} a^{i} = \beta_{1} \Bigg) \land \Bigg( \prod_{j=1}^{n} b^{j} = \alpha_{1} \Bigg) .
\end{equation}
However, in all such remaining cases, combining the above \cref{eq:remaining_cases} with the product condition for the $\alpha_{i}$ and $\beta_{j}$ given by \cref{eq:prod_rule}, and defining new symbols $\alpha_{i}^{\prime} \equiv a^{i+1} \alpha_{i+1}$ and $\beta_{j}^{\prime} \equiv b^{j+1} \beta_{j+1}$, yields
\begin{equation}
\label{eq:distinguished_prod_rule}
    \alpha_{1}^{\prime} \dots \alpha_{m-1}^{\prime} \cdot \beta_{1}^{\prime} \dots \beta_{n-1}^{\prime}
    = \prod_{i=2}^{m} a^{i} \alpha_{i} \cdot \prod_{j=2}^{n} b^{j} \beta_{j}
    = -1 .
\end{equation}

We will now assume \cref{eq:distinguished_prod_rule} to be true in order to completely bound $P(W_{X,Y})$.
Further bounding the win probability expansion of \cref{eq:winprob} by setting terms conditioned on $X = 1$ or $Y = 1$ to unity, we get
\begin{equation}
\label{eq:embedded_subgame_prob}
    P(W_{X,Y})
    \leq \frac{m+n-1}{mn}
    + \frac{(m-1)(n-1)}{mn}
    {\left[ \frac{1}{(m-1)(n-1)} \sum_{y = 2}^{n}\sum_{x = 2}^{m} P(W_{x,y} \mid X = x \cap Y = y) \right]} .
\end{equation}
Under a relabeling of the input variables, the square-bracketed terms above coincide exactly with the win probability of an \rectdim{m-1}{n-1} magic rectangle game, with its rules for row and column products specified by $\alpha_{1}^{\prime},\dots,\alpha_{m-1}^{\prime}$ and $\beta_{1}^{\prime},\dots,\beta_{n-1}^{\prime}$ respectively.
These $\alpha_{i}^{\prime}$ and $\beta_{j}^{\prime}$ specify a valid magic rectangle game since they satisfy \cref{eq:prod_rule}, as shown by \cref{eq:distinguished_prod_rule}.
Hence, we have the attainable upper bound
\begin{equation}
    \frac{1}{(m-1)(n-1)} \sum_{y = 2}^{n}\sum_{x = 2}^{m} P(W_{x,y} \mid X = x \cap Y = y)
    \leq \omega_{\Sigma}(m-1,n-1) .
\end{equation}
Combining this with \cref{eq:embedded_subgame_prob} gives the bound
\begin{equation}
    P(W_{X,Y}) \leq \bar{\omega}_{\Sigma}(m,n) ,
\end{equation}
where $\bar{\omega}_{\Sigma}(m,n)$ is defined in \cref{eq:distinguished_bound} as
\begin{equation}
    \bar{\omega}_{\Sigma}(m,n) =
    1 - \frac{(m-1)(n-1)}{mn}[1 - \omega_{\Sigma}(m-1,n-1)] .
\end{equation}
We see this has the same form as \cref{eq:value_lower_bound}.
Indeed, the proof of \cref{lem:value_lower_bound} constructs a strategy which attains this bound and is deterministic upon one input.
Finally, since
\begin{equation}
    \omega_{\Sigma}(m-1,n-1)
    \geq \omega_{L}(m-1,n-1)
    = 1 - \frac{1}{(m-1)(n-1)}
\end{equation}
for all levels of correlations $\Sigma$, \cref{eq:distinguished_bound} shows the upper bound $\bar{\omega}_{\Sigma}(m,n)$ is always at least that of $1 - (mn)^{-1} = \omega_{L}(m,n)$ found for the previously considered cases.
Therefore, $\bar{\omega}_{\Sigma}(m,n)$ represents the complete upper bound on the win probability of an \rectdim{m}{n} magic rectangle game with distinguished input and allowed nonsignaling correlation level $\Sigma$.
\end{proof}

\subsection{Performance: Noise tolerance and rates}

\begin{lemma}
\label{lem:rgen_games}
    The magic rectangle games which can be used in the $R_{gen}$ protocol are those of sizes \rectdim{2}{n} and \rectdim{3}{n} where $n \geq 2$, along with their transposed counterparts.
\end{lemma}
\begin{proof}
    We know from \cref{thm:1xn_is_classical} that \rectdim{1}{n} games do not exhibit superclassical behavior, and so cannot be used for randomness expansion.
    By \cref{thm:miller_noise}, then, we seek \rectdim{m}{n} games with $m,n \geq 2$ for which $\bar{\omega}_{Q}(m,n) < \omega_{Q}(m,n)$.
    This is clearly not the case for $m,n > 3$, since $\omega_{Q}(m,n) = 1$ for $m,n \geq 3$, and substituting this into \cref{eq:distinguished_bound} of \cref{thm:distinguished_bound} yields $\bar{\omega}_{Q}(m,n) = 1$ for $m,n > 3$.
    Thus $\omega_{Q}(m,n) = \bar{\omega}_{Q}(m,n)$ for $m,n > 3$.
    It remains to show that \rectdim{2}{n} games for $n \geq 2$ and \rectdim{3}{n} games for $n \geq 3$ can be used in $R_{gen}$.
    Then, the symmetry in $\omega_{Q}(m,n)$ provided by \cref{lem:value_symmetric} (and inherited by $\bar{\omega}_{Q}(m,n)$ through \cref{eq:distinguished_bound}) shows that games with transposed dimensions to those may also be used.

    Consider the \rectdim{2}{n} games for $n \geq 2$.
    Using \cref{thm:1xn_is_classical} in \cref{eq:distinguished_bound} gives
    \begin{equation}
        \bar{\omega}_{Q}(2,n) = 1 - \frac{1}{n} < \omega_{Q}(2,n) ,
    \end{equation}
    where the final inequality is established by comparing with \cref{eq:2xn_lower_bound}.
    Now consider the \rectdim{3}{n} games for $n \geq 3$.
    As in \cref{sec:general_2xn_games}, from \cite{arkhipov2012extending} we have the upper bound $\omega_{Q}(2,n-1) < 1$.
    Again substituting into \cref{eq:distinguished_bound} of \cref{thm:distinguished_bound}, we get
    \begin{equation}
        \bar{\omega}_{Q}(3,n) < 1 = \omega_{Q}(3,n) ,
    \end{equation}
    where the final equality uses \cref{cor:value_increasing}.
\end{proof}

For the magic rectangle games which can be used in the protocol $R_{gen}$ (shown in \cref{lem:rgen_games}), \cref{thm:miller_noise} results in a maximum noise tolerance of
\begin{equation}
\label{eq:max_noise_tol}
    \rho_{m,n}^{\text{max}} = \omega_{Q}(m,n) - \bar{\omega}_{Q}(m,n) .
\end{equation}
Furthermore, combining \cref{thm:miller_noise} with the universal lower bound of \cref{thm:miller_rate} shows that $R_{gen}$ produces (asymptotically in the number of protocol rounds) quantum-secure extractable bits at a rate of at least
\begin{equation}
    \pi(\chi) = \frac{2(\log_{2}{e})(\chi - \bar{\omega})^{2}}{r-1}
\end{equation}
per round, where $\chi \in (\bar{\omega}, \omega]$, and $r \geq 2$ is the total size of the output alphabet for the game.
According to \cref{rule:alice_prod,rule:bob_prod}, a magic rectangle game of dimension \rectdim{m}{n} has $2^{m-1} \cdot 2^{n-1}$ possible outputs.
Substituting the result of \cref{thm:distinguished_bound} for $\bar{\omega}$, this lower bound on the rate can be written for \rectdim{m}{n} magic rectangle games as
\begin{equation}
\label{eq:magic_rectangle_rate}
    \pi_{m,n}(\chi) = \frac{2(\log_{2}{e})[\chi - \bar{\omega}_{Q}(m,n)]^{2}}{2^{m+n-2} - 1} ,
\end{equation}
where $\bar{\omega}_{Q}(m,n)$ is as given in \cref{eq:distinguished_bound}.
The maximum possible lower bound that \cref{thm:miller_rate} can achieve for the rate then occurs when the score acceptance threshold is set to its maximum $\chi = \omega_{Q}(m,n)$, such that there is no tolerance to noise, and is given by
\begin{equation}
\label{eq:max_rate}
    \pi_{m,n}^{\text{max}}
    = \pi_{m,n}(\omega_{Q}(m,n))
    = \frac{2(\log_{2}{e})(\rho_{m,n}^{\text{max}})^{2}}{2^{m+n-2} - 1} .
\end{equation}
While this lower bound has the advantage that it only depends only on the dimension of the magic rectangle used, it gives rates that are far from optimal.
More practical lower bounds on the rate for the spot-checking protocol could, for example, be calculated based on the techniques of \cite{arnon2019simple}, or numerically as in \cite{brown2020framework}.

The noise tolerance for the CHSH game, or equivalently the \rectdim{2}{2} magic square game (\cref{thm:2x2_chsh_equiv}), is already known to be $(\sqrt{2} - 1)/4 \approx 10.4\%$, and this is confirmed by \cref{eq:max_noise_tol}.
Combining our characterization of magic rectangle games from \cref{sec:characterization} with the result of \cref{thm:distinguished_bound}, we summarize the performance of all viable magic rectangle games in \cref{tab:randomness_expansion_performance}.
Since the exact quantum values of the \rectdim{2}{2} and \rectdim{3}{3} games are known, inserting \cref{eq:distinguished_bound} of \cref{thm:distinguished_bound} into \cref{eq:max_noise_tol} gives exactly the optimal noise tolerance for $R_{gen}$ using the \rectdim{3}{3} game.
Hence, the \rectdim{3}{3} noise tolerance stated in \cref{tab:randomness_expansion_performance} is exact.
\begin{table}[htb]
    \centering
    \small
    \caption{
        All \rectdim{m}{n} magic rectangle games which can produce quantum-secure extractable bits in the spot-checking protocol.
        A selection of specific examples are given in the lower half of the table.
        Bounds shown for the maximum attainable noise tolerance of \rectdim{2}{n} and \rectdim{3}{n} games are given based on upper and lower bounds for the \rectdim{2}{n} quantum value (see \cref{sec:general_2xn_games}).
        Corresponding bounds are displayed for the maximal universal lower bound on the rate, as given by \cref{eq:max_rate}.
        For \rectdim{2}{2} and \rectdim{3}{3} games, upper and lower bounds coincide, so their optimal noise tolerance is exact.
        The \rectdim{3}{n} lower bounds shown for $n \geq 8$ are based on \cref{conj:2xn_upper_bound}.
        The \rectdim{2}{n} upper bounds for $n \geq 7$ are also based on \cref{conj:2xn_upper_bound}, but may be more weakly bound as in \cref{eq:weak_2xn_bounds}.
        }
    \label{tab:randomness_expansion_performance}
    \begin{threeparttable}
    \begin{tabular}{ccccc}
        \toprule
        & \multicolumn{2}{c}{Noise tolerance $\rho_{m,n}^{\text{max}}$} & \multicolumn{2}{c}{Rate bound $\pi_{m,n}^{\text{max}}$ (\si{\bit \per round})\tnotex{tn:rates}} \\
        \cmidrule(lr){2-3}\cmidrule(l){4-5}
        \rectdim{m}{n} & Upper bound & Lower bound & Upper bound & Lower bound \\
        \midrule
        \rectdim{2}{2} & $\frac{1}{4}{\left(\sqrt{2}-1\right)} \approx 10.4\%$ & $\frac{1}{4}{\left(\sqrt{2}-1\right)} \approx 10.4\%$ & $\approx 0.01031$ & $\approx 0.01031$ \\
        \addlinespace
        \rectdim{3}{3} & $\frac{1}{9}{\left(2-\sqrt{2}\right)} \approx 6.5\%$ & $\frac{1}{9}{\left(2-\sqrt{2}\right)} \approx 6.5\%$ & $\approx 0.00081$ & $\approx 0.00081$ \\
        \addlinespace
        \rectdim{2}{n}
        & $\frac{1}{2}{\left[\sqrt{1-\frac{1}{n}}-\left(1-\frac{1}{n}\right)\right]}$ & $\frac{1}{2n}{\left(\sqrt{2}-1\right)}$ & $\frac{\left(\sqrt{n(n-1)}+1-n\right)^{2}}{2(2^{n}-1)n^{2}\ln{2}}$ & $\frac{3-2\sqrt{2}}{2(2^{n}-1)n^{2}\ln{2}}$ \\
        \addlinespace
        \rectdim{3}{n}
        & $\frac{1}{3n}{\left(2-\sqrt{2}\right)}$ & $\frac{1}{3}{\Big(1-\frac{1}{n}\Big)}{\Big(1-\sqrt{1-\frac{1}{n-1}}\Big)}$ & $\frac{4(3-2\sqrt{2})}{9(2^{n+1}-1)n^{2}\ln{2}}$ & $\frac{2(n-1)(\sqrt{n-2}-\sqrt{n-1})^{2}}{9(2^{n+1}-1)n^{2}\ln{2}}$ \\
        \midrule
        \rectdim{2}{3}
        & $\frac{1}{6}{\left(\sqrt{6}-2\right)} \approx 7.5\%$ & $\frac{1}{6}{\left(\sqrt{2}-1\right)} \approx 6.9\%$ & $\approx 0.00231$ & $\approx 0.00196$ \\
        \addlinespace
        \rectdim{2}{4}
        & $\frac{1}{8}{\left(2\sqrt{3}-3\right)} \approx 5.8\%$ & $\frac{1}{8}{\left(\sqrt{2}-1\right)} \approx 5.2\%$ & $\approx 0.00065$ & $\approx 0.00052$ \\
        \addlinespace
        \rectdim{3}{4}
        & $\frac{1}{12}{\left(2-\sqrt{2}\right)} \approx 4.9\%$ & $\frac{1}{12}{\left(3-\sqrt{6}\right)} \approx 4.6\%$ & $\approx 0.00022$ & $\approx 0.00020$ \\
        \addlinespace
        \rectdim{3}{5}
        & $\frac{1}{15}{\left(2-\sqrt{2}\right)} \approx 3.9\%$ & $\frac{2}{15}{\left(2-\sqrt{3}\right)} \approx 3.6\%$ & $\approx 0.00007$ & $\approx 0.00006$ \\
        \bottomrule
    \end{tabular}
    \begin{tablenotes}
        \footnotesize
        \item[a] \label{tn:rates} These rates found from \textcite{miller2017universal} depend only on the dimension of magic rectangle game used. More practical rates could be calculated using the techniques of \cite{arnon2019simple,brown2020framework}.
    \end{tablenotes}
    \end{threeparttable}
\end{table}

It is important to note that, in \cref{tab:randomness_expansion_performance}, the upper bounds given for the noise tolerance and rate of \rectdim{2}{n} games where $n \geq 7$ are calculated based on our \cref{conj:2xn_upper_bound}, that \cref{eq:2xn_upper_bound} holds for all such $n$.
However, by trivially weakening \cref{eq:2xn_upper_bound} to $\omega_{Q}(2,n) \leq 1$, we can still find less strict upper bounds for these quantities which must hold.
Inputting this relaxation into \cref{eq:max_noise_tol,eq:max_rate}, we arrive at
\begin{equation}
\label{eq:weak_2xn_bounds}
    \rho_{2,n}^{\text{max}} \leq \frac{1}{2n} ,\quad
    \pi_{2,n}^{\text{max}} \leq \frac{\log_{2}{e}}{2n^{2}(2^{n}-1)} .
\end{equation}
These expressions are also strictly decreasing with $n$ and, for the conjectural cases of $n \geq 7$, do not exceed the upper bounds for the \rectdim{2}{3} game given in \cref{tab:randomness_expansion_performance}.

\section{Discussion}
\label{sec:conclusion}

In this work, we defined a class of nonlocal games which we called ``magic rectangles'', since they are natural generalizations of the \textcite{mermin1990simple,peres1990incompatible} magic square.
As a first point for future work, it would be interesting to further generalize our games to the multipartite scenario, in which players would output by filling $(d-1)$-dimensional slices of a ``magic hyperrectangle'' of $d$ dimensions.
By characterizing a suitable generalization of this kind, it may also be possible to identify other well-known nonlocal games as special cases.

Our main results can be divided into two parts.
Firstly, we obtained a full characterization of magic rectangle games with respect to the winning probabilities of quantum and classical strategies.
Secondly, we focused on one important application, namely certified randomness expansion; we demonstrated how a complete characterization can be used to explore the potential for device-independent protocols of all the family of nonlocal games we introduced.
We will discuss these two parts separately, giving future directions for each.

\subsection{Characterization}

We obtained a complete characterization of magic rectangle games.
We have shown that \rectdim{1}{n} games cannot exhibit superclassical behavior.
Moreover, any magic rectangle game of at least size \rectdim{3}{3} can be won with certainty using quantum or stronger correlations.
For these games, the interesting properties of strong contextuality and implementation with only Clifford computations of the regular magic square game are preserved.
We have also shown that the special case of dimension \rectdim{2}{2} is identical to the CHSH game, which is well studied and does not exhibit the aforementioned properties.

Finally, the class of \rectdim{2}{n} games for $n \geq 3$ is seen to exhibit the richest behavior: there do not exist perfect quantum winning strategies for these games, however, we have shown superclassical lower bounds on their optimal success probabilities using quantum correlations.
We have also given numerical upper bounds on quantum win probabilities for these games with small $n$, and conjectured a closed-form expression extending to all $n$.
An interesting consequence of our analysis of \rectdim{2}{n} magic rectangle games is that they provide examples of nonlocal games that can be won with certainty using NPA level 1 correlations, and yet for which no quantum (or numerically almost quantum) winning strategy exists (see also \cref{sec:2x3_at_npa1} for an example).

\paragraph{Future works.}
An interesting future direction is to closer examine this special class of \rectdim{2}{n} magic rectangles.
The problem of finding optimal quantum values is still an open question, where the possibilities that they coincide with our lower bounds, upper bounds, or something between all have interesting implications.
In the first case, optimal strategies could be implemented using CHSH sub-games.
Games of the third case would outperform the CHSH game while also exhibiting a separation between the quantum and almost quantum sets.
We believe the second case, in which the quantum and almost quantum sets coincide for each magic rectangle, to be the most likely.
This would provide further evidence of the naturality of almost quantum correlations.
Once specific strategies (for games beyond CHSH) have been obtained, one could directly see how these perform for various device-independent cryptographic primitives or self-testing.

\subsection{Certified randomness expansion}
\label{sec:randomness_discussion}

The optimal noise tolerance of an \rectdim{m}{n} magic rectangle game for certified randomness expansion in the spot-checking protocol is fully determined by the difference of the optimal quantum win probability $\omega_Q(m,n)$ and the optimal quantum win probability with distinguished input $\bar{\omega}_{Q}(m,n)$.
In \cref{thm:distinguished_bound}, we relate $\bar{\omega}_{Q}(m,n)$ with $\omega_{Q}(m-1,n-1)$, and given that we have characterized the quantum win probabilities for magic rectangle games of all dimensions in \cref{thm:characterization}, we can obtain the noise tolerance of all magic rectangle games (\cref{tab:randomness_expansion_performance}).
Specifically, the noise tolerance of an \rectdim{m}{n} is given as the difference between its quantum value, and the corresponding value of the \rectdim{(m-1)}{(n-1)} game extended to dimension \rectdim{m}{n} by including in each of its outputs a deterministic entry.
It follows that only magic rectangle games of dimension \rectdim{2}{n} and \rectdim{3}{n}, with $n \geq 2$ can be used for certified randomness expansion (larger rectangle games fail, since the games can be won with certainty even with a distinguished input).
Moreover, we can also see from \cref{tab:randomness_expansion_performance} that the most robust game turns out to be the \rectdim{2}{2} magic square game (which we showed is equivalent to the CHSH game).
The values given for general \rectdim{2}{n} and \rectdim{3}{n} games are strictly decreasing with $n$ and, furthermore, of these only the \rectdim{2}{2} and \rectdim{2}{3} games outperform the noise tolerance and rate bound given for the \rectdim{3}{3} game.

From the equivalence with the CHSH game, optimal strategies for the \rectdim{2}{2} game can be implemented using only a single Bell state shared between the players, whereas all known implementations of optimal strategies for the \rectdim{3}{3} game require a system of at least two Bell states.
However, implementations of certain winning \rectdim{3}{3} strategies may still be advantageous, for example in cases where physical limitations on the quantum devices dictate certain additional constraints (such as requiring the use of only Clifford gates), or in the context of self-testing (where the use of pairs of Bell states enables parallel self-testing).

\paragraph{Future works.}
An important remaining question is that of the optimal rates that one can achieve with magic rectangle games.
Since we showed that, in terms of noise tolerance, the optimal game coincides with the CHSH game, analysis of the rates has already been done extensively.
However, it is still an interesting problem to obtain rates for all the games (whether this is because one is interested in a specific game, or because a protocol may provide better rates with worse noise tolerance---something conceivably possible).

Note that in \cref{tab:randomness_expansion_performance} we do give some rates for all the different games.
\Cref{thm:miller_rate} directly relates noise tolerance to a lower bound on the rate of randomness expansion, which we can (and do) use to directly obtain indicative rates (\cref{tab:randomness_expansion_performance} last column).
However, we would like to stress that the rates obtained from this expression (unlike our noise tolerance analysis) are far from optimal.
More practical rates can be calculated, for example, by referring to the techniques outlined in \cite{arnon2019simple}, or numerically as in \cite{brown2020framework}.
To obtain these improved rates requires an involved, case-by-case analysis that treats each magic rectangle game separately, something that is sensible to do if one is interested in a given game, and is left for future publications.

\section*{Acknowledgments}
We would like to thank Matty J. Hoban for useful discussions.
S.A.A. gratefully acknowledges EPSRC studentship funding under grant number EP/R513209/1.

\appendix

\section{Winning 2-by-3 games at NPA level 1}
\label{sec:2x3_at_npa1}

Consider the \rectdim{2}{3} magic rectangle game in which entries to the first column are required to have a negative product, and all other row and column products are required to be positive.
That is, the \rectdim{2}{3} game specified by the parameters $(\beta_{1},\beta_{2},\beta_{3}) = (-,+,+)$ and $(\alpha_{1},\alpha_{2}) = (+,+)$ satisfying \cref{def:magic_rectangle}.
In order to write our strategy more easily, in \cref{tab:2x3_biject} we introduce a more concise alphabet for the inputs and outputs of the game.
\begin{table}[htb]
    \centering
    \small
    \caption{
        The natural alphabets $\mathscr{A}$ and $\mathscr{B}$ defined here denote new notation for the natural alphabets of the \rectdim{2}{3} magic rectangle game under consideration, with parameters $(\alpha_{1},\alpha_{2}) = (+,+)$ and $(\beta_{1},\beta_{2},\beta_{3}) = (-,+,+)$.
        Elements of each alphabet have the form of input/output pairs for each player, with the input written first.
        }
    \label{tab:2x3_biject}
    \begin{tabular}{cccc}
        \toprule
        $\mathscr{A}_{2 \times 3}$ & $\mathscr{A}$ & $\mathscr{B}_{2 \times 3}$ & $\mathscr{B}$ \\
        \cmidrule(r){1-2}\cmidrule(l){3-4}
        $(1, (+,+,+))$ & $(1,1)$ & $(1, (+,-)^{\transpose})$ & $(1,1)$ \\
        $(1, (+,-,-))$ & $(1,2)$ & $(1, (-,+)^{\transpose})$ & $(1,2)$ \\
        $(1, (-,+,-))$ & $(1,3)$ & $(2, (+,+)^{\transpose})$ & $(2,1)$ \\
        $(1, (-,-,+))$ & $(1,4)$ & $(2, (-,-)^{\transpose})$ & $(2,2)$ \\
        $(2, (+,+,+))$ & $(2,1)$ & $(3, (+,+)^{\transpose})$ & $(3,1)$ \\
        $(2, (+,-,-))$ & $(2,2)$ & $(3, (-,-)^{\transpose})$ & $(3,2)$ \\
        $(2, (-,+,-))$ & $(2,3)$ & & \\
        $(2, (-,-,+))$ & $(2,4)$ & & \\
        \bottomrule
    \end{tabular}
\end{table}

Under the new notation defined in \cref{tab:2x3_biject}, the success probability of a behavior $P(a,b \mid x,y)$ where $(x,a) \in \mathscr{A}$ and $(y,b) \in \mathscr{B}$ is
\begin{equation}
\label{eq:2x3_winprob}
\begin{split}
    p = \tfrac{1}{6}
        [ & P(1,1 \mid 1,1) + P(2,1 \mid 1,1) + P(3,2 \mid 1,1) + P(4,2 \mid 1,1) \\
    {}+{} & P(1,1 \mid 1,2) + P(2,2 \mid 1,2) + P(3,1 \mid 1,2) + P(4,2 \mid 1,2) \\
    {}+{} & P(1,1 \mid 1,3) + P(2,2 \mid 1,3) + P(3,2 \mid 1,3) + P(4,1 \mid 1,3) \\
    {}+{} & P(1,2 \mid 2,1) + P(2,2 \mid 2,1) + P(3,1 \mid 2,1) + P(4,1 \mid 2,1) \\
    {}+{} & P(1,1 \mid 2,2) + P(2,2 \mid 2,2) + P(3,1 \mid 2,2) + P(4,2 \mid 2,2) \\
    {}+{} & P(1,1 \mid 2,3) + P(2,2 \mid 2,3) + P(3,2 \mid 2,3) + P(4,1 \mid 2,3) ] .
\end{split}
\end{equation}
We now state a behavior, achievable using NPA level 1 correlations, for which the win probability $p$ of \cref{eq:2x3_winprob} is unity.
This behavior is defined via the matrices
\begin{subequations}
\label{eq:2x3_behavior}
\begin{align}
    (P(a,b \mid 1,1))_{a,b} &= \frac{1}{4}
        \begin{pmatrix}
            1 & 0 \\
            1 & 0 \\
            0 & 1 \\
            0 & 1
        \end{pmatrix} , \\
    (P(a,b \mid 2,1))_{a,b} &= \frac{1}{4}
        \begin{pmatrix}
            0 & 1 \\
            0 & 1 \\
            1 & 0 \\
            1 & 0
        \end{pmatrix} , \\
    (P(a,b \mid 1,2))_{a,b} = (P(a,b \mid 2,2))_{a,b} &= \frac{1}{4}
        \begin{pmatrix}
            1 & 0 \\
            0 & 1 \\
            1 & 0 \\
            0 & 1
        \end{pmatrix} , \\
    (P(a,b \mid 1,3))_{a,b} = (P(a,b \mid 2,3))_{a,b} &= \frac{1}{4}
        \begin{pmatrix}
            1 & 0 \\
            0 & 1 \\
            0 & 1 \\
            1 & 0
        \end{pmatrix} .
\end{align}
\end{subequations}
Indeed, the behavior defined by \cref{eq:2x3_behavior} admits an NPA hierarchy level 1 certificate, given by the matrix
\begin{equation}
    \Gamma = \frac{1}{8}
        \begin{pmatrix}
            8  &  2  &  2  &  2  &  2  &  2  &  2  &  4  &  4  &  4 \\
            2  &  2  &  0  &  0  &  1  & -1  &  1  &  2  &  2  &  2 \\
            2  &  0  &  2  &  0  & -1  &  1  &  1  &  2  &  0  &  0 \\
            2  &  0  &  0  &  2  &  1  &  1  &  1  &  0  &  2  &  0 \\
            2  &  1  & -1  &  1  &  2  &  0  &  0  &  0  &  2  &  2 \\
            2  & -1  &  1  &  1  &  0  &  2  &  0  &  0  &  0  &  0 \\
            2  &  1  &  1  &  1  &  0  &  0  &  2  &  2  &  2  &  0 \\
            4  &  2  &  2  &  0  &  0  &  0  &  2  &  4  &  2  &  2 \\
            4  &  2  &  0  &  2  &  2  &  0  &  2  &  2  &  4  &  2 \\
            4  &  2  &  0  &  0  &  2  &  0  &  0  &  2  &  2  &  4
        \end{pmatrix} .
\end{equation}
By \cref{cor:value_unique}, we thus have that $\omega_{1}(2,3) = 1$.
Therefore, by \cref{cor:value_increasing}, $\omega_{1}(2,n) = 1$ for all $n \geq 3$.

\printbibliography

\end{document}